\documentclass[11pt]{article}
\usepackage{times}

\usepackage{subcaption}
\usepackage{caption}
\usepackage[paperwidth=199.8mm,paperheight=297mm,centering,hmargin=15mm,vmargin=2cm]{geometry}
\usepackage[dvipsnames]{xcolor}
\usepackage{framed}
\definecolor{shadecolor}{rgb}{0.9,0.9,0.9}
\usepackage{mathtools}
\usepackage{amsmath}
\usepackage[shortlabels]{enumitem}
\usepackage[titletoc,title]{appendix}
\usepackage{graphicx,epic,eepic,epsfig,amsmath,latexsym,amssymb,verbatim,color}
\usepackage{dsfont}

\usepackage{float}
\usepackage{tikz}
\usepackage{hyperref}
\hypersetup{colorlinks=true,citecolor=blue,linkcolor=blue,filecolor=blue,urlcolor=blue,breaklinks=true}

\usepackage[marginal]{footmisc}
\usepackage{url}
\usepackage{theorem}
\newtheorem{definition}{Definition}
\newtheorem{proposition}[definition]{Proposition}
\newtheorem{lemma}[definition]{Lemma}

\newtheorem{theorem}[definition]{Theorem}

\usepackage{colortbl}
\usepackage{pifont}
\definecolor{Gray}{gray}{0.92}
\definecolor{Gray2}{gray}{0.75}
\definecolor{maroon}{cmyk}{0,0.87,0.68,0.32}
\usepackage{booktabs} 
\usepackage{multirow}
\def\squareforqed{\hbox{\rlap{$\sqcap$}$\sqcup$}}
\def\qed{\ifmmode\squareforqed\else{\unskip\nobreak\hfil
\penalty50\hskip1em\null\nobreak\hfil\squareforqed
\parfillskip=0pt\finalhyphendemerits=0\endgraf}\fi}
\def\endenv{\ifmmode\;\else{\unskip\nobreak\hfil
\penalty50\hskip1em\null\nobreak\hfil\;
\parfillskip=0pt\finalhyphendemerits=0\endgraf}\fi}
\newenvironment{proof}{\noindent \textbf{{Proof~} }}{\hfill $\blacksquare$ \vspace{0.2cm}}

\newcounter{remark}
\newenvironment{remark}[1][]{\refstepcounter{remark}\par\medskip\noindent%
\textbf{Remark~\theremark #1} }{\medskip}

\newcounter{example}

\mathchardef\ordinarycolon\mathcode`\:
\mathcode`\:=\string"8000
\def\vcentcolon{\mathrel{\mathop\ordinarycolon}}
\begingroup \catcode`\:=\active
  \lowercase{\endgroup
  \let :\vcentcolon
  }

\usepackage{graphicx}
\usepackage{xcolor}

\RequirePackage[framemethod=default]{mdframed}
\newmdenv[skipabove=7pt,
skipbelow=7pt,
backgroundcolor=darkblue!15,
innerleftmargin=5pt,
innerrightmargin=5pt,
innertopmargin=5pt,
leftmargin=0cm,
rightmargin=0cm,
innerbottommargin=5pt,
linewidth=1pt]{tBox}

\newmdenv[skipabove=7pt,
skipbelow=7pt,
backgroundcolor=blue2!25,
innerleftmargin=5pt,
innerrightmargin=5pt,
innertopmargin=5pt,
leftmargin=0cm,
rightmargin=0cm,
innerbottommargin=5pt,
linewidth=1pt]{dBox}
\newmdenv[skipabove=7pt,
skipbelow=7pt,
backgroundcolor=darkkblue!15,
innerleftmargin=5pt,
innerrightmargin=5pt,
innertopmargin=5pt,
leftmargin=0cm,
rightmargin=0cm,
innerbottommargin=5pt,
linewidth=1pt]{sBox}
\definecolor{darkblue}{RGB}{0,76,156}
\definecolor{darkkblue}{RGB}{0,0,153}
\definecolor{blue2}{RGB}{102,178,255}
\definecolor{darkred}{RGB}{195,0,0}

\newcommand{\nc}{\newcommand}
\nc{\rnc}{\renewcommand}
\nc{\beg}{\begin{equation}}
\nc{\eeq}{{\end{equation}}}
\nc{\beqa}{\begin{eqnarray}}
\nc{\eeqa}{\end{eqnarray}}
\nc{\lbar}[1]{\overline{#1}}
\nc{\bra}[1]{\langle#1|}
\nc{\ket}[1]{|#1\rangle}
\nc{\ketbra}[2]{|#1\rangle\!\langle#2|}
\nc{\braket}[2]{\langle#1|#2\rangle}

\nc{\proj}[1]{| #1\rangle\!\langle #1 |}
\nc{\avg}[1]{\langle#1\rangle}
\nc{\Rank}{\operatorname{Rank}}
\nc{\smfrac}[2]{\mbox{$\frac{#1}{#2}$}}
\nc{\tr}{\operatorname{Tr}}
\nc{\ox}{\otimes}

\nc{\cA}{{\cal A}}
\nc{\cB}{{\cal B}}
\nc{\cC}{{\cal C}}
\nc{\cD}{{\cal D}}
\nc{\cE}{{\cal E}}
\nc{\cF}{{\cal F}}
\nc{\cG}{{\cal G}}
\nc{\cH}{{\cal H}}
\nc{\cI}{{\cal I}}
\nc{\cJ}{{\cal J}}
\nc{\cK}{{\cal K}}
\nc{\cL}{{\cal L}}
\nc{\cM}{{\cal M}}
\nc{\cN}{{\cal N}}
\nc{\cO}{{\cal O}}
\nc{\cP}{{\cal P}}
\nc{\cQ}{{\cal Q}}
\nc{\cR}{{\cal R}}
\nc{\cS}{{\cal S}}
\nc{\cT}{{\cal T}}
\nc{\cV}{{\cal V}}
\nc{\cX}{{\cal X}}
\nc{\cY}{{\cal Y}}
\nc{\cZ}{{\cal Z}}
\nc{\cW}{{\cal W}}
\nc{\csupp}{{\operatorname{csupp}}}
\nc{\qsupp}{{\operatorname{qsupp}}}
\nc{\var}{{\operatorname{var}}}
\nc{\rar}{\rightarrow}
\nc{\lrar}{\longrightarrow}
\nc{\polylog}{{\operatorname{polylog}}}
\nc{\1}{{\mathds{1}}}
\nc{\wt}{{\operatorname{wt}}}
\nc{\av}[1]{{\left\langle {#1} \right\rangle}}
\nc{\supp}{{\operatorname{supp}}}

\def\ve{\varepsilon}

\nc{\RR}{{{\mathbb R}}}
\nc{\CC}{{{\mathbb C}}}
\nc{\FF}{{{\mathbb F}}}
\nc{\NN}{{{\mathbb N}}}
\nc{\ZZ}{{{\mathbb Z}}}
\nc{\PP}{{{\mathbb P}}}
\nc{\QQ}{{{\mathbb Q}}}
\nc{\UU}{{{\mathbb U}}}
\nc{\EE}{{{\mathbb E}}}
\nc{\id}{{\operatorname{id}}}

\nc{\CHSH}{{\operatorname{CHSH}}}

\nc{\be}{\begin{equation}}
\nc{\ee}{{\end{equation}}}
\nc{\bea}{\begin{eqnarray}}
\nc{\eea}{\end{eqnarray}}
\nc{\<}{\langle}
\rnc{\>}{\rangle}
\nc{\rU}{\mbox{U}}

\nc{\ob}[1]{#1}

\nc{\SEP}{{\text{\rm SEP}}}
\nc{\NS}{{\text{\rm NS}}}
\nc{\LOCC}{{\text{\rm LOCC}}}
\nc{\PPT}{{\text{\rm PPT}}}
\nc{\EXT}{{\text{\rm EXT}}}
\nc{\Sym}{{\operatorname{Sym}}}

\nc{\ERLO}{{E_{\text{r,LO}}}}
\nc{\ERLOCC}{{E_{\text{r,LOCC}}}}
\nc{\ERPPT}{{E_{\text{r,PPT}}}}
\nc{\ERLOCCinfty}{{E^{\infty}_{\text{r,LOCC}}}}
\nc{\Aram}{{\operatorname{\sf A}}}

\usepackage{tikz}
\RequirePackage{doi}
\usepackage{hyperref}
\hypersetup{colorlinks=true,citecolor=blue,linkcolor=blue,filecolor=blue,urlcolor=blue,breaklinks=true}

\makeatletter
\def\grd@save@target#1{%
  \def\grd@target{#1}}
\def\grd@save@start#1{%
  \def\grd@start{#1}}
\tikzset{
  grid with coordinates/.style={
    to path={%
      \pgfextra{%
        \edef\grd@@target{(\tikztotarget)}%
        \tikz@scan@one@point\grd@save@target\grd@@target\relax
        \edef\grd@@start{(\tikztostart)}%
        \tikz@scan@one@point\grd@save@start\grd@@start\relax
        \draw[minor help lines,magenta] (\tikztostart) grid (\tikztotarget);
        \draw[major help lines] (\tikztostart) grid (\tikztotarget);
        \grd@start
        \pgfmathsetmacro{\grd@xa}{\the\pgf@x/1cm}
        \pgfmathsetmacro{\grd@ya}{\the\pgf@y/1cm}
        \grd@target
        \pgfmathsetmacro{\grd@xb}{\the\pgf@x/1cm}
        \pgfmathsetmacro{\grd@yb}{\the\pgf@y/1cm}
        \pgfmathsetmacro{\grd@xc}{\grd@xa + \pgfkeysvalueof{/tikz/grid with coordinates/major step}}
        \pgfmathsetmacro{\grd@yc}{\grd@ya + \pgfkeysvalueof{/tikz/grid with coordinates/major step}}
        \foreach \x in {\grd@xa,\grd@xc,...,\grd@xb}
        \node[anchor=north] at (\x,\grd@ya) {\pgfmathprintnumber{\x}};
        \foreach \y in {\grd@ya,\grd@yc,...,\grd@yb}
        \node[anchor=east] at (\grd@xa,\y) {\pgfmathprintnumber{\y}};
      }
    }
  },
  minor help lines/.style={
    help lines,
    step=\pgfkeysvalueof{/tikz/grid with coordinates/minor step}
  },
  major help lines/.style={
    help lines,
    line width=\pgfkeysvalueof{/tikz/grid with coordinates/major line width},
    step=\pgfkeysvalueof{/tikz/grid with coordinates/major step}
  },
  grid with coordinates/.cd,
  minor step/.initial=.2,
  major step/.initial=1,
  major line width/.initial=2pt,
}
\makeatother

\usepackage{thmtools}
\usepackage{thm-restate}
\usepackage{etoolbox}
\makeatletter
\def\problem@s{}
\newcounter{problems@cnt}

\newcommand{\allproblems}{\problem@s}
\makeatother

\usepackage{authblk} 
\usepackage{tcolorbox}
\usepackage{relsize}
\usepackage{graphicx}
\nc{\MIO}{\text{\rm MIO}}
\nc{\DIO}{\text{\rm DIO}}
\nc{\SIO}{\text{\rm SIO}}
\nc{\IO}{\text{\rm IO}}
\nc{\PIO}{\text{\rm PIO}}
\nc{\DIIO}{\text{\rm DIIO}}
\nc{\LICC}{\text{\rm LICC}}  
\nc{\LQICC}{\text{\rm LQICC}} 
\nc{\SI}{\text{\rm SI}}     
\nc{\SQI}{\text{\rm SQI}}   
\nc{\QIP}{\text{\rm QIP}}   

\newcommand{\cU}{\mathcal U}

\newcommand{\sbar}{\;\rule{0pt}{9.5pt}\middle|\;}
\newcommand{\channel}[1]{{\cal C}(#1)}
\definecolor{beamer}{rgb}{0.2,0.2,0.7} 

\newcommand{\rel}{\middle\|}

\begin{document}

\title{\Large \textbf{Finite Block Length Analysis on Quantum Coherence Distillation\\ and Incoherent Randomness Extraction}}

\author[1,2,3,4]{Masahito Hayashi~\thanks{masahito@math.nagoya-u.ac.jp}}
\affil[1]{\small Graduate School of Mathematics, Nagoya University, Nagoya, 464-8602, Japan}
\affil[2]{Shenzhen Institute for Quantum Science and Engineering, \protect\\ Southern University of  Science and Technology, Shenzhen, 518055, China}
\affil[3]{Center for Quantum Computing, Peng Cheng Laboratory, Shenzhen 518000, China}
\affil[4]{Centre for Quantum Technologies, National University of Singapore, 3 Science Drive 2, 117542, Singapore}
\author[5]{\ Kun Fang~\thanks{fangfred11@gmail.com}}
\affil[5]{Institute for Quantum Computing, University of Waterloo, Waterloo, Ontario N2L 3G1, Canada}
\author[6,2]{\ Kun Wang~\thanks{nju.wangkun@gmail.com}}
\affil[6]{Institute for Quantum Computing, Baidu Research, Beijing 100193, China}

\date{\today}
\maketitle

\vspace{-0.5cm}
\begin{abstract}



We give the \emph{first} systematic study on the second order asymptotics of the operational task of coherence distillation with and without assistance.
In the unassisted setting,
we introduce a variant of randomness extraction framework where free incoherent operations are allowed before the incoherent measurement and the randomness extractors. We then show that the maximum number of random bits extractable
from a given quantum state is precisely equal to the maximum number of coherent bits that can be distilled from
the same state. This relation enables us to derive tight second order expansions of both tasks
in the independent and identically distributed setting.
Remarkably, the incoherent operation classes that can empower coherence distillation for generic states all admit the same second order expansions, indicating their operational equivalence for coherence distillation in both asymptotic and large block length regime.
We then generalize the above line of research to the assisted setting,
arising naturally in bipartite quantum systems where Bob distills
coherence from the state at hand, aided by the benevolent Alice possessing the other
system. More precisely,
we introduce a new assisted incoherent randomness extraction task and
establish an exact relation between this task and the assisted coherence distillation.
This strengthens the one-shot relation in the unassisted setting
and confirms that this cryptographic framework indeed offers a new perspective to the study of
quantum coherence distillation. Likewise, this relation yields second order characterizations to
the assisted tasks.
As by-products, we show the strong converse property of the aforementioned
tasks from their second order expansions.
\end{abstract}

{
  \hypersetup{linkcolor=black}
  \tableofcontents
}


\section{Introduction}



A central problem of a general resource theory~\cite{chitambar2018quantum} is \emph{resource distillation}:
the process of extracting canonical units of resources from a given quantum state
using a choice of free operations.
The usual asymptotic approach to studying problems in quantum information theory is to assume
that there is an {unbounded} number of independent and identically distributed (i.i.d.) copies of
a source state available and that the error measure asymptotically goes to zero.
However, for many practical applications there are natural restrictions on the code length imposed,
for example, by limitations on how much quantum information can be processed coherently.
Therefore it is crucial to go beyond the asymptotic treatment and understand the intricate
tradeoff between different operational parameters of concern.

In general, suppose a quantity of interest is given by $R(\rho^{\ox n},\ve)$ which is a
function of the sequence of states $\{\rho^{\ox n}\}_{n \in \NN}$ and the error threshold $\ve$.
The typical information-theoretical study focuses on finding the asymptotic
rate $R_1(\rho):= \lim_{\ve \to 0}\lim_{n\to \infty} \frac1n R(\rho^{\ox n},\ve)$.
Together with a strong converse property, this is equivalent to expanding $R(\rho^{\ox n},\ve) = n R_1 + o(n)$
where $R_1$ is called the \emph{first order coefficient}. An estimation to the order $o(n)$ is usually
unsatisfactory in the case with limited resources, motivating us to further investigate the \emph{second order expansion} -- a refined estimation of $R(\rho^{\ox n},\ve)$
to the order $o(\sqrt{n})$. More precisely, we aim to find an
expansion $R(\rho^{\ox n},\ve) = n R_1 + \sqrt{n} R_2 + o(\sqrt{n})$ where $R_2$ is called the
 \emph{second order asymptotics}. To achieve so, we shall first consider the \emph{one-shot} scenario where the source is characterized by a single instance of \emph{unstructured} quantum state,
 and then the \emph{second order} setting in which only a \emph{finite} number of copies of i.i.d. state is given.
 
 The significance of second order expansions is multifold. First, second order expansions provide a useful approximation for finite block length $n$, refining optimal rates that typically correspond to the first order asymptotics in asymptotic expansions. Second, they determine the rate of convergence to the first order asymptotics, analogous to the relation between the Central Limit Theorem and the Berry-Esseen Theorem, as the latter determines the rate of convergence in the former. Finally, second order expansions can be used to derive the strong converse property, an information-theoretic property that rules out a possible tradeoff between the transformation error and the rate of resulting resource of a protocol.

In this paper we focus on the resource theory of quantum coherence~\cite{streltsov_2017}
and investigate in depth the second order asymptotics of the coherence distillation task under various settings. In the following two paragraphs, we summarize individually the state-of-the-art research progress
of both the unassisted and assisted coherence distillation tasks, point out the major concerns,
and present shortly our strategy to solve these concerns and the obtained results.

\paragraph{Unassisted coherence distillation and randomness extraction} Quantum coherence is a physical resource that is essential for various information processing tasks~\cite{hillery2016coherence,coles2016numerical,streltsov_2016,anshu2018quantum,diaz2018,lostaglio2015description,Frowis2011}. A series of efforts have been devoted to building this resource theory in recent years~\cite{aberg2006quantifying,gour_2008,levi_2014,baumgratz2014quantifying,streltsov_2017}, characterizing in particular the state transformations and operational uses of coherence in fundamental resource manipulation protocols~\cite{winter_2016,chitambar_2016-3,streltsov_2016,chitambar_2016-2,regula2018one,fang2018probabilistic,zhao_2018}. The task of coherence distillation in the asymptotic scenario has been first investigated in~\cite{winter_2016} and has been recently completed in~\cite{lami2019completing}. In spite of their theoretical importance, the asymptotic assumptions become unphysical in reality due to our limited access to a finite number of copies of a given state, making it necessary to look at non-asymptotic regimes. The first step in this direction is to consider the one-shot setting that distills coherence from a single instance of the prepared state. Such a scenario has been investigated in~\cite{regula2018one} and has been mostly completed in~\cite{zhao_2018}. These works estimate the one-shot distillable coherence under different free operations by their corresponding one-shot entropies. The one-shot entropies most accurately describe the operational quantity, yet they tend to be difficult to calculate for large systems, even for the independent and identically distributed (i.i.d.) case. This motivates further investigations of second order expansions.

The usual approach to deriving the second order expansion of an information task is to combine the one-shot entropy bounds on the information quantity and the second order expansion of the corresponding entropies (e.g.~\cite{hayashi2008second,tomamichel2013hierarchy,li2014second,datta2014second,tomamichel2016quantum,fang2019non,wang2019converse}). However, as second order expansions have a strong dependence on the error parameter $\ve$, the existing one-shot entropy bounds on distillable coherence~\cite{regula2018one,zhao2019one} are insufficient to get a tight second order expansion. That is, the second order coefficients in the expansion of the one-shot entropy lower and upper bounds are often mismatched.
To solve this, we introduce a variant of randomness extraction framework in the context of quantum coherence theory \cite{Yuan,hayashi2018secure}
and build an exact connection of this task with coherence distillation. Such a connection provides us a new perspective to the study of distillation process.
Finally, expanding a one-shot entropy lower bound on the extractable randomness and a one-shot entropy upper bound on the distillable coherence, we obtain the desired second order expansion.
The exact one-shot relation between randomness extraction and coherence distillation builds a bridge between two seemly different information tasks, providing new perspectives to the study of both problems. Moreover, our second order expansions initiate \emph{the first} large block length analysis in quantum coherence theory, filling an important gap in the literature.


\paragraph{Assisted coherence distillation and randomness extraction}
Chitambar~\textit{et al.}~\cite{chitambar_2016-3}
originally proposed the assisted distillation task,
arising naturally in bipartite quantum systems in which Bob aims to distill
coherent bits from the state at his hand, aided by the friendly Alice who possesses the other
system~\cite{divincenzo1998entanglement,gregoratti2003quantum,hayden2004correcting,smolin2005entanglement,winter2005environment,buscemi2005inverting,buscemi2007channel,dutil2010assisted,buscemi2013general,karumanchi2016classical,karumanchi2016quantum,lami2020assisted}.
Streltsov \textit{et al.}~\cite{streltsov2017towards} pushed forward this task by enlarging the set of
free operation classes, mimicking the bipartite operation hierarchy well studied in the entanglement resource
theory~\cite{horodecki2009quantum}. They coined this the resource theory of coherence in distributed scenarios.
Subsequent works have been carried out both on
the theoretical~\cite{zhao2017coherence,zhao20191} and
experimental~\cite{wu2017experimentally,wu2018experimental,wu2020quantum} directions.
Recently, Regula \textit{et al.}~\cite{regula2018nonasymptotic} and Vijayan \textit{et al.}~\cite{vijayan2018one}
independently studied the one-shot assisted coherence distillation via different approaches, enhancing the
asymptotic results~\cite{chitambar2016assisted,streltsov2017towards} to the one-shot realm.
The obtained results are not satisfactory due to the following reasons.
Firstly, they assumed the bipartite state pre-shared to be pure \emph{a priori}. Practically, it is common that Bob
holds an extension, rather than a purification, of Bob's state. Secondly, they considered a special class of free
operations in which Bob perform measurements and Bob perform conditional incoherent operations. It would be
meaningful to explore the power of other free classes in the distributed incoherent operation hierarchy. Lastly but
most importantly, the obtained one-shot bounds cannot lead to second order asymptotics.
As we have argued, the seek for second order expansion is necessary since it not only provides
a useful approximation for the finite block length, but also determines the rate of convergence
to its first order coefficient.

Our approach towards these concerns is largely inspired by the ideas from the unassisted scenario.
More specifically, we propose a variant of randomness
extraction framework within the context of distributed quantum coherence theory, which we call the assisted
incoherent randomness extraction. We  establish an equivalence relation between this task and the assisted coherence
distillation in the one-shot regime. It uplifts the relation between coherence distillation and incoherent randomness
extraction in the unassisted setting to the assisted scenario.
Then, we make use of this relation to draw a complete
characterization on these two tasks by proving an one-shot lower bound on the assisted extractable randomness
and an one-shot upper bound on the assisted distillable coherence,
with matching dependence on the error threshold $\ve$ in the two bounds.
This bestows us with desired second order expansion for both tasks.

\begin{table}[H]
\centering
\renewcommand{\arraystretch}{1.3}
\begin{tabular}{@{}ccc@{}}\toprule[1.5pt]
    \textbf{Assistance}
  & \textbf{Characterization}
  & \textbf{Reference} \\\midrule[0.3pt]
    \multirow{5}{*}{Unassisted case}
  & one-shot equivalence
  & \multirow{2}{*}{Theorem~\ref{thm: one-shot relation}} \\
    {}
  & $C_{d,\cO}^\ve(\rho_B) = \ell_{\cO}^{\ve}(\rho_B)$
  & {} \\\cmidrule[0.3pt](l){2-3}
    {}
  & second order expansion
  & \multirow{3}{*}{Theorem~\ref{thm: second order}} \\
    {}
  & $C_{d,\cO}^{\ve}(\rho_B^{\ox n}) = \ell_{\cO}^{\ve}(\rho_B^{\ox n}) =$
  & {} \\
    {}
  & $n D(\rho_B\|\Delta(\rho_B)) + \sqrt{n V(\rho_B\|\Delta(\rho_B))}\, \Phi^{-1}(\ve^2) + O(\log n)$
  & {} \\\midrule[0.5pt]
     \multirow{5}{*}{Assisted case}
  & one-shot equivalence
  & \multirow{2}{*}{Theorem~\ref{thm:equivalence}} \\
    {}
  & $C_{d,\QIP}^\ve(\rho_{AB}) = \ell_{\QIP}^{\ve}(\rho_{AB})$
  & {} \\\cmidrule[0.3pt](l){2-3}
    {}
  & second order expansion
  & \multirow{3}{*}{Theorem~\ref{thm: second order EA}} \\
    {}
  & $C_{d,\QIP}^{\ve}(\rho_{AB}^{\ox n}),\;
        \ell_{\cF}^{\ve}(\rho_{AB}^{\ox n}) =$
  & {} \\
    {}
  & $nD\left(\rho_{AB}\rel\Delta_B(\rho_{AB})\right)
  + \sqrt{nV\left(\rho_{AB}\rel\Delta_B(\rho_{AB})\right)}\,\Phi^{-1}(\varepsilon^2) + O(\log n)$
  & {} \\
\bottomrule[1.5pt]
\end{tabular}
\caption{List of results with references obtained in this work regarding both unassisted and assisted coherence distillation and incoherent randomness extraction tasks. Here $\cO\in\{\MIO,\DIO,\IO,\DIIO\}$ is a single partite free operation class and $\cF\in\{\LICC,\LQICC,\SI,\SQI,\QIP\}$ is a bipartite free operation class.
}
\label{tbl:list-of-results}
\end{table}

\paragraph{Outline and main contributions}
The main contributions of this paper are listed in Table~\ref{tbl:list-of-results} for references and can be
summarized as follows:
\begin{itemize}
\item In Section~\ref{sec: coherence and randomness}, we first propose a variant of randomness extraction framework in
the context of quantum coherence theory, and then establish an \emph{exact} relation between the task of randomness
extraction and the task of quantum coherence distillation in the one-shot regime. More precisely, we show that the
maximum number of secure randomness bits ($\ell_{\cO}^\ve$) extractable from a given state is equal to the maximum
number of coherent bits ($C_{d,\cO}^\ve$) that can be distilled from the same state. That is, for any quantum state
$\rho$, error tolerance $\ve \in [0,1]$, it holds $C_{d,\cO}^\ve(\rho) = \ell_{\cO}^\ve(\rho)$, where free operation
class $\cO \in \{\MIO,\DIO,\IO,\DIIO\}$ whose definitions can be found in Section~\ref{sec: resource theory of quantum
coherence}. We further give one-shot achievable and converse bounds for $\ell_{\cO}^\ve$ as well as $C_{d,\cO}^\ve$ in terms of hypothesis testing relative entropy, through which we  get the second order expansion of our information tasks.
  \item In Section~\ref{sec:assisted},
        we first propose the assisted incoherent randomness extraction task
        within the quantum coherence theory in distributed scenario,
        and then set up an equivalence relation between the assisted coherence
        distillation and the assisted incoherent randomness extraction in the one-shot regime. That is, we show
        that the maximum number of coherent bits ($C_{d,\QIP}^\varepsilon$) that can be distilled from a bipartite
        quantum state $\rho_{AB}$ is equal to the maximum number of secure randomness bits ($\ell_{\QIP}^\varepsilon$)
        of the same state, where $\QIP$ is the set of quantum-incoherent state preserving operations (cf.
        Section~\ref{sec:distributed}). Finally,
        we prove one-shot achievability bounds for $\ell_{\cF}^\varepsilon$
        and one-shot converse bounds for $C_{d,\cF}^\varepsilon$, where $\cF\in\{\LICC,\LQICC,\SI,\SQI,\QIP\}$,
        in terms of the hypothesis testing relative entropy.
        These, together with the established one-shot equivalence relation,
        yield an one-shot characterization to these two quantities.
        Invoking the second order expansion of the hypothesis testing relative entropy,
        we are able to get the second order expansion of these rates.
\item In Section~\ref{sec: strong converse}, we show the strong converse property of
      all these explored tasks -- coherence distillation and incoherent randomness extraction and their assisted
      versions -- by using the established second order expansions.
\item In Section~\ref{sec: tripartite state}, we establish a series of relations among various entropic quantities
      evaluated on the dephased tripartite state $\Delta_B(\Psi_{RAB})$, where $\ket{\Psi}_{RAB}$ is a purification
      of $\rho_{AB}$. These relations are essential in the second order analysis and may find applications in other
      quantum information processing tasks.
\item In Appendix~\ref{sec:alternative formulation},
      we conceive an alternative formulation of the assisted incoherent randomness extraction. The advantage
      of this alternative is that we can establish an equivalence relation between the assisted coherence distillation
      ($C_{d,\cF}^\ve$) and the alternative assisted incoherent randomness extraction ($\widehat{\ell}_{\cF}^\ve$) for
      all free operation classes under consideration.
\end{itemize}

\section{Preliminaries}

In this section we define several quantities and set the notation that will be used throughout this paper. We label different physical systems by capital Latin letters (e.g. $A,C,L$). We often use these labels as subscripts to guide the reader by indicating which system a mathematical object belongs to. We drop the subscripts when they are evident in the context of an expression (or if we are not talking about a specific system). The corresponding Hilbert spaces of these physical systems are denoted as $\cH_A,\cH_C,\cH_L$ respectively. The corresponding alphabet sets are denoted by the same letters in mathcal font (e.g. $\cA,\cC,\cL$). For example, $\cA := \{1,2,\cdots, |A|\}$ where $|A|$ is the dimension of Hilbert space $\cH_A$. Let $\{\ket{a}\}_{a\in \cA}$ be the computational basis on Hilbert space $\cH_A$. The set of positive semidefinite operators on $\cH_A$ is denoted as $\cP(A)$. The set of quantum states, which are positive semidefinite operators with unit trace, on $\cH_A$ is denoted as $\cS(A)$.  Denote the completely mixed state on $\cH_A$ as $\pi_A$. The identity operator and the identity map are denoted as $\1$ and $\id$ respectively. A quantum operation $\Lambda_{A\to C}$ is a completely positive trace-preserving (CPTP) map from $\cS(A)$ to $\cS(C)$. All logarithms in this work are taken base two.

For any $\rho,\sigma \in \cP$, the purified distance $P$ is defined in terms of the generalized quantum fidelity $F$ as $P(\rho,\sigma) := \sqrt{1-F(\rho,\sigma)^2}$ with $F(\rho,\sigma):=\|\sqrt{\rho}\sqrt{\sigma}\|_1 + \sqrt{(1-\tr \rho)(1-\tr \sigma)}$~\cite{tomamichel2015quantum}.
For any $\rho \in \cS$ and $\sigma \in \cP$, their quantum hypothesis testing relative entropy is defined as
$D_H^\ve(\rho\|\sigma):= -\log \min\{\tr M \sigma: \tr M \rho \geq 1-\ve , 0 \leq M \leq \1\}$~\cite{wang2012one}.
The smooth max-relative entropy is defined as $D_{\max}^\ve(\rho\|\sigma) := \min_{P(\tilde \rho,\rho) \leq \ve} \inf\{\lambda: \widetilde \rho \leq 2^\lambda \sigma\}$~\cite{datta2009min}. The second order expansions of quantum hypothesis testing relative entropy~\cite{tomamichel2013hierarchy,li2014second} and smooth max-relative entropy~\cite{tomamichel2013hierarchy} are, respectively, given by
\begin{align}
     D_H^\ve(\rho^{\ox n}\|\sigma^{\ox n}) & = n D(\rho\|\sigma) + \sqrt{n V(\rho\|\sigma)}\, \Phi^{-1}(\ve) + O(\log n),\label{eq: DH second order} \\
     D_{\max}^\ve(\rho^{\ox n}\|\sigma^{\ox n}) & = n D(\rho\|\sigma) - \sqrt{nV(\rho\|\sigma)}\, \Phi^{-1}(\ve^2) + O(\log n),\label{eq: Dmax second order}
\end{align}
where $D(\rho\|\sigma) := \tr [\rho (\log \rho - \log \sigma)]$ is the quantum relative entropy, $V(\rho\|\sigma) := \tr [\rho (\log \rho - \log \sigma)^2] - D(\rho\|\sigma)^2$ is the quantum information variance and $\Phi^{-1}$ is the inverse of the cumulative distribution function of a standard normal random variable.
The logarithms are in base $2$ throughout this paper.

For any $\sigma\in\cP$, we denote the number of distinct non-zero eigenvalues of $\sigma$ by
$\nu(\sigma)$. Let $\lambda_{\max}$ and $\lambda_{\min}$ be the maximum and minimum non-zero eigenvalues of $\sigma$,
respectively. Set $\lambda(\sigma):=\log \lambda_{\max}(\sigma) - \log \lambda_{\min}(\sigma)$. The following $\theta$
function is commonly used in one-shot information quantities:
\begin{align}\label{eq: definition of theta function}
    \theta(\sigma) := \min\{2\lceil \lambda(\sigma)\rceil, \nu(\sigma)\}.
\end{align}

\section{Quantum coherence distillation and incoherent randomness extraction}
\label{sec: coherence and randomness}

In this section we first review the resource theory of quantum coherence and the operational task of quantum coherence distillation. We then introduce a variant of randomness extraction framework in the context of quantum coherence theory.

\subsection{Resource theory of quantum coherence}
\label{sec: resource theory of quantum coherence}

The resource theory of coherence consists of the following ingredients~\cite{baumgratz2014quantifying}: the set of
\emph{free states} and the set of \emph{free operations}, that is, a set of quantum operations that do not generate
coherence. The free states, so-called incoherent states, are the quantum states which are diagonal in a given reference
orthonormal basis $\{
\ket{b}\}_{b\in\cB}$.
We will use
$\Delta_B (\cdot):=\sum_ {b\in\cB} \ket{b}\bra{b} \cdot \ket{b}\bra{b}$ to denote the diagonal map (completely
dephasing channel) in this basis~\footnote{In this section we assume that Bob instead of the commonly used Alice is the main party of all the operations, thus the symbol of our notations starts from the letter $B$. We do so to keep consistent with the later assisted scenario in which Alice serves as an assistance.}. Then the set of incoherent states is denoted as $\cI(B):= \{\rho \in \cS(B): \rho =
\Delta_B(\rho)\}$. For convenience, we will also use the cone of diagonal positive semidefinite matrices, which is
denoted as $\cI^{**}(B):= \{X \in \cP(B): X = \Delta_B(X)\}$. The maximal resource state on $\cH_B$ is the
\emph{maximally coherent state (MCS)} $\ket{\Psi_B} := 1/\sqrt{|B|}\sum_{b=1}^{|B|} \ket{b}$ with dimension $|B|$.
Denote its density operator as $\Psi_B:= \ket{\Psi_B}\bra{\Psi_B}$.  The resource theory of coherence is known not to
admit a unique physically-motivated choice of allowed free operations
\cite{winter_2016,chitambar_2016,marvian_2016,vicente_2017,streltsov_2017}. The relevant choices of free operations
that we will focus on are: \emph{maximally incoherent operations (MIO)}~\cite{aberg2006quantifying}, defined to be all
operations $\Lambda$ such that $\Lambda(\rho) \in \cI$ for every $\rho \in \cI$;  \emph{dephasing-covariant incoherent
operations (DIO)}~\cite{chitambar_2016,marvian_2016}, which are maps $\Lambda$ such that $\Delta \circ \Lambda =
\Lambda \circ \Delta$; \emph{incoherent operations (IO)}~\cite{baumgratz2014quantifying}, which admit a set of
incoherent Kraus operators $\{K_l\}$ such that ${K_l \rho K_l^\dagger} \in \cI^{**}$ for all $l$ and $\rho\in \cI$; the
intersection of $\IO$ and $\DIO$ is denoted as $\DIIO:=\DIO \cap \IO$~\cite{zhao2019one}. Another two classes of free
operations commonly studied are \emph{strictly incoherent operations (SIO)}~\cite{winter_2016} and \emph{physically
incoherent operations (PIO)}~\cite{chitambar_2016}. We do not investigate further details of SIO and PIO, as it has
been recently shown that quantum coherence is generically undistillable under these two
classes~\cite{lami2019generic,lami2019completing}. Finally, the inclusion relations between free operation classes can
be summarized as  $\DIIO \subsetneq \IO \subsetneq \MIO$, $\DIIO \subsetneq \DIO \subsetneq \MIO$, while $\IO$ and
$\DIO$ are not contained by each other.

\subsection{Framework of quantum coherence distillation}

The task of \emph{coherence distillation} aims to transform a given quantum state $\rho_B$ to a maximally coherent state $\Psi_C$ such that the obtained maximally coherent state has dimension as large as possible and that the transformation error is within a given threshold. More formally, for any free operation class $\cO$, any given state $\rho_B \in \cS(B)$ and error threshold $\ve \in [0,1]$, the \emph{one-shot distillable coherence} is defined as
\begin{align}\label{eq: definition one shot coherence}
    C_{d,\cO}^{\ve}(\rho_B):= \max \big\{\log |C|: P(\Lambda_{B\to C}(\rho_B),\Psi_C) \leq \ve, \Lambda \in \cO\big\}.
\end{align}
{Note that some previous works (e.g.~\cite{regula2018one,fang2018probabilistic,zhao2019one}) estimate the performance of distillation by the error criterion $P(\Lambda_{B\to C}(\rho_B),\Psi_C) \leq \sqrt{\ve}$. Here we use the definition in~\eqref{eq: definition one shot coherence} for convenience.}

\subsection{Framework of incoherent randomness extraction}\label{sec:incoherent randomness extraction}

\begin{figure}[!hbtp]
\centering
\includegraphics[width=0.7\textwidth]{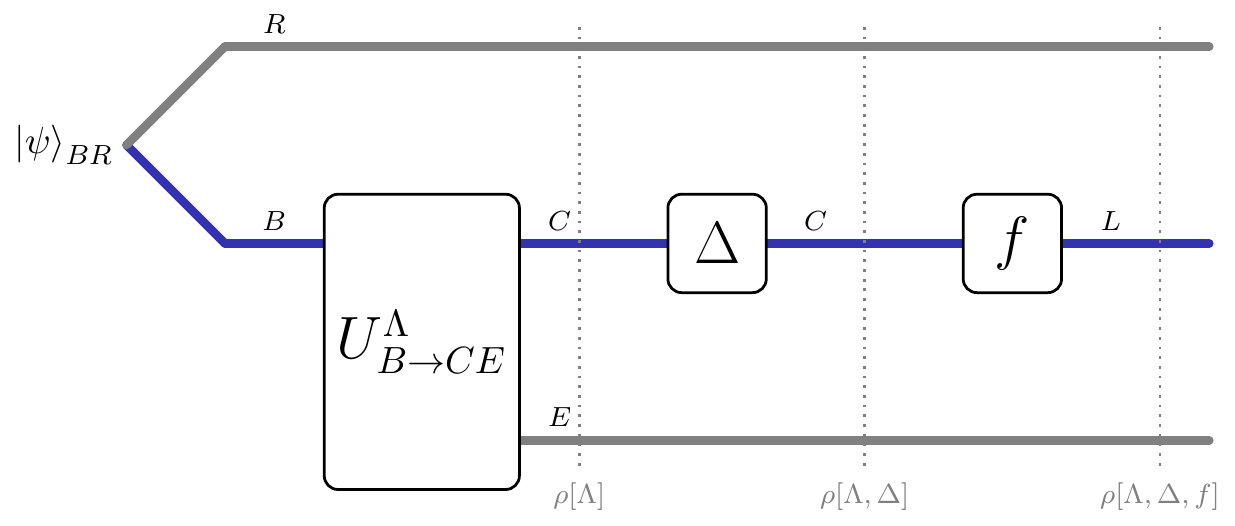}
\caption{\small Schematic diagram of an incoherent randomness extraction protocol given by $(\Lambda,\Delta,f)$. $\ket{\psi}_{BR}$ is a purification of $\rho_B$. $U^\Lambda_{B\to CE}$ 
is the Stinespring representation of $\Lambda_{B\to C} \in \cO$. $\Delta$
is a completely dephasing channel and $f$ is a hash function from alphabet $\cC$ to $\cL$. $\rho[\Lambda]$, $\rho[\Lambda,\Delta]$ and $\rho[\Lambda,\Delta,f]$ are respectively the output states in each step of the protocol. The systems in blue belong to Bob and the systems in red belong to Eve.}
\label{fig: protocol of incoherent randomness extraction 3}
\end{figure}

The task of \emph{incoherent randomness extraction} aims to obtain as many random bits as possible in Bob's laboratory that is secure from the possible adversary Eve. A general \emph{incoherent randomness extraction protocol} is characterized by a triplet $(\Lambda,\Delta,f)$, where $\Lambda$ is an incoherent operation in a certain class, $\Delta$ is a completely dephasing channel and $f$ is a hash function. A detailed procedure of randomness extraction by $(\Lambda,\Delta,f)$ is depicted in Figure~\ref{fig: protocol of incoherent randomness extraction 3}. Here, we assume that Eve has unlimited power in her system and all the information of Eve about Bob’s systems is encoded in a purification. That is, for any given quantum state $\rho_B$ held by Bob, we denote its purification as~\footnote{Note that $\ket{\psi_b}$ are not necessarily orthogonal to each other.}
\begin{align}
\ket{\psi}_{BR}:= \sum_{b \in \cB} \sqrt{p_b} \ket{b}_B\ket{\psi_b}_R    \quad \text{with}\quad \tr_R \ket{\psi}\bra{\psi}_{BR} = \rho_B,
\end{align}
where $R$ is the reference system held by Eve. Then the extraction protocol follows by three steps:
\begin{enumerate}
\item  Bob first performs a free operation $\Lambda_{B\to C} \in \cO$ on his part of the system. If he uses a quantum operation whose final state is always a specific incoherent state, say the maximally mixed state $\pi_C$,
the resulting conditional entropy equals $\log |C|$, which increases unlimitedly as $|C|$ increases.
To avoid such a trivial advantage for Bob, similar to the study of quantum key distribution \cite{tomamichel2012tight,shor2000simple,hayashi2006practical} and private capacity in quantum Shannon theory~\cite{li2009private}, we assume that the environment system $\cH_{E}$ of the free operation $\Lambda$ is also controlled by Eve. This is because it is not easy to exclude the possibility that Eve accesses a system that interacts with  Bob's operation. Hence Eve has control over two systems $\cH_R$ and $\cH_{E}$ in total.  To cover such a worst scenario, we consider the Stinespring representation $U^\Lambda$ of $\Lambda$, where $U^\Lambda$ is the isometry from  $\cH_B$ to $\cH_C \otimes \cH_E$.~\footnote{Note that a free operation does not necessarily admit a free dilation~\cite{chitambar_2016}. Thus $U^\Lambda$ is not necessarily incoherent though $\Lambda$ is free.}
After the action of $\Lambda$, the total output state is a pure state
\begin{align}
{\rho}[\Lambda]_{CER} := U^\Lambda \proj{\psi} (U^\Lambda)^\dagger.
\end{align}

\item Next, Bob performs an incoherent measurement, with respect to the computational basis, on his part of the state. The output state is then given by
\begin{align}
    \rho[\Lambda, \Delta]_{CER}:= \Delta_C(\rho[\Lambda]_{CER}).
\end{align}

\item Finally, a hash function $f$ is applied on his part of the system to extract the randomness that is secure from Eve.
For any deterministic function $f: \cC \to \cL$, and any classical-quantum (CQ) state $\rho_{CR} = \sum_{c\in \cC} t_c \ket{c}\bra{c}_C \ox \rho_{R|c}$, denote $\rho_{f(C)R} := \sum_{c\in \cC} t_c \ket{f(c)}\bra{f(c)}_L \ox \rho_{R|c}$.
Then the output state in the final step is given by
\begin{align}
    \rho[\Lambda,\Delta,f]_{LER}:= \rho[\Lambda,\Delta]_{f(C)ER}.
\end{align}
\end{enumerate}

To quantify the security of   randomness in a quantum state $\rho_{BR}$ with respect to the reference system $\cH_R$, we employ the following measure:
\begin{align}
d_{sec}({\rho}_{BR}|R):=
\min_{\sigma_R \in \cS(R)} P( {\rho}_{BR}, \pi_B \otimes {\sigma}_{R}).
\label{min1}
\end{align}
This measure quantifies the closeness of a given quantum state $\rho_{BR}$ to an ideal state $\pi_B$ which contains perfect randomness and is decoupled from the reference system $\cH_R$. Based on this security measure, the \emph{one-shot extractable randomness under given free operation} $\Lambda \in \cO$ is defined as
\begin{align}
    \ell_{\Lambda}^\ve(\rho_B):= \max_f \big\{\log |L|: d_{sec}(\rho[\Lambda,\Delta,f]_{LER}|ER) \leq \ve\big\}.
\end{align}
The \emph{one-shot extractable randomness under free operation class} $\cO$ is then defined by optimizing all possible choices of operation $\Lambda$ in $\cO$:
\begin{align}
    \ell_{\cO}^\ve(\rho_B):= \max_{\Lambda \in \cO} \ell_{\Lambda}^\ve(\rho_B).
\end{align}

Note that the identity map $\id$ is always free in coherence theory. Thus $(\id,\Delta,f)$ is a valid incoherent
randomness extraction protocol for any $f$, which was studied in~\cite[Section III]{tomamichel2013hierarchy}. That is,
Bob directly performs an incoherent measurement $\Delta$ on his given state $\rho_B$. In this case, the environment
system $\cH_E$ reduces to trivial and system $\cH_C = \cH_B$.  It has been shown in~\cite[Theorem
8]{tomamichel2013hierarchy} that for any $\eta \in (0,\ve]$,
\begin{align} \label{eq: hashing lemma}
H_{\min}^{\ve-\eta}(B|R)_{\sigma} + 4\log \eta-3
\le
\ell_{\id}^\ve(\rho_B)
\le
H_{\min}^{\ve}(B|R)_{\sigma},
\end{align}
where
\begin{align}\label{eq:conditional min}
     H_{\min}^\ve(B|R)_{\rho} := \max_{\omega_R \in \cS(R)} - D_{\max}^\ve(\rho_{BR}\|\1_B \ox \omega_R)
\end{align}
is the conditional min-entropy and $\sigma_{BR}:= \rho[\id,\Delta]$ is the dephased classical-quantum state in the protocol.

\begin{remark}
On the one hand, the randomness extraction protocol without using incoherent operations (e.g. the one considered in~\cite[Figure 1.(b)]{zhao2019one}) is too restrictive, as such a framework does not make good use of free resources at hand. On the other hand, an extraction protocol that does not consider Eve's attack on the free operation $\Lambda$ is too trivial because Bob can generate an arbitrary amount of randomness by using a free replacer channel $\Lambda(\cdot)=\pi_C$. Hence, the setup in Figure~\ref{fig: protocol of incoherent randomness extraction 3} contributes to a reasonable randomness extraction framework in the context of quantum coherence theory.
\end{remark}

\begin{remark}
Note that performing incoherent unitary operations in the first step does not make any difference with the protocol by identity map $\id$. This justifies our consideration of general incoherent operations. More precisely, for any incoherent unitary $U$, it holds
$\ell_{\id}^\ve(\rho_B) = \ell_{\cU}^\ve(\rho_B)$ with $\cU(\cdot):= U(\cdot) U^\dagger$.
To see this, recall that any incoherent unitary on $\cH_B$ can be written as $U_B=\sum_{b\in \cB} e^{i\theta_b}\ket{g(b)}\bra{b}$ with a permutation $g$ and phase factors $e^{i\theta_b}$~\cite[Section II.A.2]{streltsov_2017}. Then a direct calculation gives that $\rho[\cU,\Delta,f] = \sum_{b \in \cB} p_b \ket{f(g(b))}\bra{f(g(b))}\ox \ket{\psi_b}\bra{\psi_b} = \rho[\id,\Delta,f\circ g]$, implying the equivalence of extraction protocols $(\cU,\Delta,f)$ and $(\id,\Delta,f\circ g)$. Then $\ell_{\id}^\ve(\rho_B) = \ell_{\cU}^\ve(\rho_B)$ follows by definition.
\end{remark}

\begin{remark}
The randomness extraction framework proposed here is also closely related to the one in~\cite{hayashi2018secure}.
But we should note the following subtle differences: (i) the independence of the resulting randomness is quantified using the trace distance in~\cite{hayashi2018secure} instead of the purified distance we use in~\eqref{min1}. Though the trace distance can give us the nice property of universal composability (see e.g.~\cite{Renner2005}),
the choice of purified distance in~\eqref{min1} is crucial for obtaining the exact relation between coherence distillation and randomness extraction in the next section, which then becomes a key ingredient to proving the second order expansions;
(ii)  while the paper~\cite{hayashi2018secure} has a discussion on the large block length regime, its analysis is focused on the exponential decreasing rate
for the amount of the leaked information $\ve$,
but we will put more focus on the rate of extractable randomness in the one-shot and large block length regime with constant $\ve$.

The references \cite{Yuan} and \cite[Section VI]{hayashi2018secure} also address the randomness extraction via incoherent operations but with additional constraints on Eve, which are different from our setting here.

\end{remark}

\subsection{Relation between coherence distillation and randomness extraction}
\label{sec: one shot relation}

In this section we present an exact connection between incoherent randomness extraction and coherence distillation, the proof of which showcases a nice one-to-one correspondence between protocols in these tasks.

We first present a method of constructing a coherence distillation protocol from a given 
randomness extraction protocol with the same performance, which will be used in showing the 
one-shot connection. We note that this method has essentially been 
proposed in the proof of~\cite[Theorem 4]{zhao2019one}, using the trace distance in 
the security measure.
Whereas in our formulation, the purified distance is preferred and is crucial in establishing tight second order asymptotics. 

\begin{proposition}\label{prop: IO protocol}
For any quantum state $\rho_B\in\cS(B)$, error tolerance $\ve \in [0,1]$ and incoherent randomness extraction protocol $(\id,\Delta,f)$ such that $d_{sec}(\rho[\id,\Delta,f]_{LR}|R) \leq \ve$, there exists a quantum operation $\Gamma: \cS(B) \to \cS(L)$ such that $\Gamma \in \DIIO$ and $P(\Gamma_{B\to L}(\rho_B), \Psi_L ) \leq \ve$.
\end{proposition}

The proof can be found in Appendix~\ref{sec:proposition-one-proof}.

\begin{theorem}[Exact one-shot connection]\label{thm: one-shot relation}
For any quantum state $\rho_B \in \cS(B)$, error tolerance $\ve \in [0,1]$ and free operation class $\cO \in \{\MIO,\DIO,\IO,\DIIO\}$, the following equation holds
\begin{align}\label{eq: one-shot relation}
     C_{d,\cO}^\ve(\rho_B) = \ell_{\cO}^{\ve}(\rho_B).
\end{align}
\end{theorem}

\begin{remark}
    Recall that there is an exact characterization of one-shot distillable coherence under MIO and DIO operations~\cite[Proposition 2]{regula2018one}. Together with the above~\eqref{eq: one-shot relation}, we have the following chain of equalities:
    \begin{align}\label{eq: one shot tmp1}
        C_{d,\MIO}^\ve(\rho_B) = C_{d,\DIO}^\ve(\rho_B) = \ell_{\MIO}^\ve(\rho_B) = \ell_{\DIO}^\ve(\rho_B) = \min_{\substack{X = \Delta(X)\\ \tr X = 1}} D_H^\ve(\rho_B\|X) - \delta,
    \end{align}
    where the minimum is taken over all Hermitian operators $X$ on $\cH_B$ satisfying the conditions and $\delta \geq 0$ is the least number such that the solution corresponds to the logarithm of an integer.
\end{remark}

\begin{remark}
A one-shot relation between distillable coherence and extractable randomness has appeared in~\cite[Equation (80)]{zhao2019one}. Unlike the precise equation in~\eqref{eq: one-shot relation}, the relation in~\cite{zhao2019one} is given in the form of one-shot lower and upper bounds with unmatched error dependence and additional correction terms. However, the clean form in~\eqref{eq: one-shot relation} plays a pivotal role in deriving the second order expansions where the error dependence matters.
\end{remark}

\noindent \textbf{[Proof of Theorem~\ref{thm: one-shot relation}]}
We first show the direction $\ell^{\ve}_{\cO}(\rho_B) \ge C_{d,\cO}^{\ve}(\rho_B).$ Denote $C_{d,\cO}^{\ve}(\rho_B) = \log |C|$ and suppose that this rate is achieved by a free operation $\Lambda: \cS(B) \to \cS(C)$ such that $P(\Lambda (\rho_B),\Psi_C ) \le \ve$. Consider a randomness extraction protocol $(\Lambda,\Delta,\id)$. Note that $\rho[\Lambda]_{CER}$ is a purification of $\Lambda(\rho_B)$. By Uhlmann's theorem~\cite{uhlmann1976transition} there exists an extension of $\Psi_C$, denoted as $\Psi_C \ox \sigma_{ER}^*$, such that $P(\Lambda(\rho_B),\Psi_C) = P(\rho[\Lambda]_{CER},\Psi_C\ox \sigma_{ER}^*)$. Then we have
\begin{align}
d_{sec}(\rho[\Lambda,\Delta,\id]_{CER}|ER) = & \min_{\sigma_{ER} \in \cS(ER)}  P\Big(\Delta_C (\rho[\Lambda]_{CER}),
\pi_C \otimes \sigma_{ER} \Big) \\
\le & P\Big(\Delta_C (\rho[\Lambda]_{CER}), \pi_C \otimes \sigma_{ER}^* \Big)
\\
= & P\Big(\Delta_C (\rho[\Lambda]_{CER}), \Delta_C(\Psi_C \otimes \sigma_{ER}^*) \Big) \\
\le & P\Big(\rho[\Lambda]_{CER},\Psi_C \otimes \sigma_{ER}^* \Big) \\
= & P\big(\Lambda (\rho_{A}), \Psi_C\big) \\
\le & \ve,
\end{align}
where the second equality follows by $\Delta_C(\Psi_C) = \pi_C$, the second inequality follows by the data-processing inequality of purified distance, the third equality follows from the assumption of $\sigma^*_{ER}$. Thus we know that $\log |C|$ is an achievable randomness extraction rate, which implies $\ell_{\cO}^\ve(\rho_B) \geq \log|C| = C_{d,\cO}^\ve(\rho_B)$.

For the other direction, we denote $\ell^{\ve}_{\cO}(\rho_B) = \log |L| $ and suppose that this rate is achieved by an extraction protocol $(\Lambda,\Delta,f)$ with $\Lambda \in \cO$. Notice that applying the protocol $(\Lambda,\Delta,f)$ on quantum state $\rho_B$ is the same as applying a protocol $(\id,\Delta,f)$ on $\Lambda(\rho_B)$ with purification $\rho[\Lambda]_{CER}$ and reference system $ER$. By Proposition~\ref{prop: IO protocol} there exists a quantum operation $\Gamma: \cS(C) \to \cS(L)$ such that $\Gamma \in \DIIO$ and $P(\Gamma(\Lambda(\rho_B)),\Psi_L) \leq \ve$. Since $\Gamma \in \DIIO \subseteq \cO$ and $\Lambda \in \cO$, we have $\Gamma\circ \Lambda \in \cO$ and this operation achieves the distillation rate $\log |L|$. This implies $C_{d,\cO}^{\ve}(\rho_B) \geq \log |L| = \ell_{\cO}^{\ve}(\rho_B)$ and completes the proof.
\hfill $\blacksquare$

\subsection{Second order analysis}
\label{sec: second order}

In this section we discuss the second order expansions of distillable coherence and extractable randomness. For this, we first show a one-shot characterization of distillable coherence by the hypothesis testing relative entropy.

\begin{proposition}[One-shot characterization]\label{prop: one shot estimation}
    For any quantum state $\rho_B \in \cS(B)$, free operation class $\cO \in \{\MIO,\DIO,\IO,\DIIO\}$, error tolerance $\ve \in (0,1)$ and  $0<\eta<\ve$, $0 < \delta < \min\{({\ve}-\eta)^2/3,1-({\ve}-\eta)^2\}$, it holds
    \begin{align}\label{eq: one shot estimation}
        D_H^{({\ve}-\eta)^2-2\delta}(\rho_B\|\Delta(\rho_B)) - c(\rho_B,\ve,\delta,\eta) \leq C_{d,\cO}^{\ve}(\rho_B) \leq D_H^{\ve^2}(\rho_B\|\Delta(\rho_B)),
    \end{align}
    where $c(\rho_B,\ve,\delta,\eta)=\log \theta(\rho_B) + \log \theta(\Delta(\rho_B)) + \log (({\ve}-\eta)^2-\delta)-\log (\delta^5 \eta^4 ({\ve}-\eta)^2 (1-({\ve}-\eta)^2+\delta)) + 11$.
\end{proposition}
\begin{proof}
The upper bound follows from a known result in~\cite[Proposition 2]{regula2018one} (see also~\eqref{eq: one shot tmp1}).
Choosing a feasible solution $X$ as $\Delta(\rho)$, we have
$C_{d,\cO}^{\ve}(\rho_B) \leq C_{d,\MIO}^{\ve}(\rho_B) \leq  D_H^{\ve^2}(\rho_B\|\Delta(\rho_B))$
where the first inequality follows by the fact that $\cO \subseteq \MIO$. As for the lower bound, we have $C_{d,\cO}^{\ve}(\rho_B) = \ell_{\cO}^\ve(\rho_B) \geq \ell_{\id}^\ve(\rho_B) \geq H_{\min}^{\ve-\eta}(B|R)_{\sigma} + 4\log \eta-3$ where the first equality follows from Theorem~\ref{thm: one-shot relation}, the first inequality follows since $\id \in \cO$, the second inequality follows from~\eqref{eq: hashing lemma}. By~\eqref{eq: Hmin DH} of Proposition~\ref{prop:relations} we can further lower bound the conditional min-entropy by
the hypothesis testing relative entropy. This gives us the one-shot lower bound stated in~\eqref{eq: one shot
estimation}.
\end{proof}

Now we are ready to present our main result in this section.

\begin{theorem}[Second order expansion]\label{thm: second order}
    For any quantum state $\rho_B \in \cS(B)$, error tolerance $\ve \in (0,1)$ and free operation class $\cO \in \{\MIO,\DIO,\IO,\DIIO\}$, the following second order expansions hold
    \begin{align}\label{eq: second order result thm}
        C_{d,\cO}^{\ve}(\rho_B^{\ox n}) =
        \ell_{\cO}^{\ve}(\rho_B^{\ox n}) = n D(\rho_B\|\Delta(\rho_B)) + \sqrt{n V(\rho_B\|\Delta(\rho_B))}\, \Phi^{-1}(\ve^2) + O(\log n),
    \end{align}
    where $\Phi^{-1}$ denotes the inverse of the cumulative distribution function of a standard normal random variable.
\end{theorem}

\begin{proof}
The first equality follows from the one-shot relation in Theorem~\ref{thm: one-shot relation}. We now prove that the second order expansion holds for $C_{d,\cO}^{\ve}(\rho_B^{\ox n})$. This can be seen as a direct consequence of expanding the one-shot characterization in Proposition~\ref{prop: one shot estimation}. Given the i.i.d. state $\rho_B^{\ox n}$, we have 
\begin{align}
C_{d,\cO}^{\ve}(\rho_B^{\ox n}) \leq D_{H}^{\ve^2}(\rho_B^{\ox n}\|\Delta(\rho_B)^{\ox n}).
\end{align}
Expanding the r.h.s. via formula~\eqref{eq: DH second order} we have the second order upper bound. On the other hand, we have 
\begin{align}\label{eq: second order tmp1}
C_{d,\cO}^{\ve}(\rho_B^{\ox n}) \geq D_H^{({\ve}-\eta)^2-2\delta}(\rho_B^{\ox n}\|\Delta(\rho_B)^{\ox n}) - c(\rho_B^{\ox n},\ve,\delta,\eta).
\end{align}
By definition~\eqref{eq: definition of theta function}, $\theta(\sigma^{\ox n}) \leq 2\lceil\lambda(\sigma^{\ox n})\rceil =2\lceil n\lambda(\sigma) \rceil$ which scales at most linearly in $n$. Choosing $\eta$ and $\delta$ proportional to $1/\sqrt{n}$, we know that the correction term $c(\rho^{\ox n},\ve,\delta,\eta) \in O(\log n)$ and $\Phi^{-1}(({\ve}-\eta)^2-2\delta) = \Phi^{-1}(\ve^2) + O(1/\sqrt{n})$. Thus expanding the r.h.s. of~\eqref{eq: second order tmp1} via formula~\eqref{eq: DH second order} leads to the second order lower bound which matches exactly with the upper bound.
\end{proof}

\begin{remark}
From Proposition~\ref{prop:relations} in Section~\ref{sec: tripartite state}, the first and second order asymptotics can also be written as conditional entropy $H(B|R)_{\sigma_{BR}}$ and conditional information variance $V(B|R)_{\sigma_{BR}}$ respectively, where $\sigma_{BR}= \rho[\id,\Delta]$ is the dephased classical-quantum state from the randomness extraction protocol.
\end{remark}

\begin{remark}\label{rm: randomness extraction advantage}
Comparing the second order expansion of $\ell_{\id}^{\ve}(\rho_B^{\ox n})$ in~\cite[Corollary 16]{tomamichel2013hierarchy} and the result above, we can conclude that a general incoherent randomness extraction protocol $(\Lambda,\Delta,f)$ has no advantage over the protocol $(\id,\Delta,f)$ in the sense that they lead to the same first order asymptotics~\cite{hayashi2018secure} and the second order asymptotics of extractable randomness. However, this does not rule out a possible advantage in the third or higher order terms.
\end{remark}

\begin{remark}
The distillable coherence under MIO/DIO/IO/DIIO not only have the same first order asymptotics as observed in~\cite{winter_2016,regula2018one,zhao2019one} but also have the same second order asymptotics, indicating that they are equivalently powerful for coherence distillation in the large block length regime. The same argument goes to the incoherent randomness extraction.
\end{remark}

\begin{remark}
For any quantum state $\rho_B = \sum_{i,j} \rho_{ij} \ket{i}\bra{j}_B$ written in the computational basis, we can
assign it to a bipartite maximally correlated state $\rho_{\text{mc}}:= \sum_{i,j} \rho_{ij} \ket{i}\bra{j}_A \ox
\ket{i}\bra{j}_B$. The second order expansion of  distillable entanglement of $\rho_{\text{mc}}^{\ox n}$ under local
operations and classical communication (LOCC) is also given by $n D(\rho\|\Delta(\rho)) + \sqrt{n
V(\rho\|\Delta(\rho))}\, \Phi^{-1}(\ve^2) + O(\log n)$~\cite[Proposition 10]{fang2019non}. Together with the result in
Theorem~\ref{thm: second order}, the coincidence of these second order expansions leads to a new evidence to the
long-standing conjecture (see e.g.~\cite[Section II.D]{streltsov_2017}) that any incoherent operation acting on a state
$\rho_B$ is equivalent to a LOCC operation acting on the associated maximally correlated state $\rho_{\text{mc}}$.
\end{remark}

\begin{remark}
Compared with the one-shot estimation in~\cite[Equations (37,46,47)]{zhao2019one}, we can verify that their upper and lower bounds on the one-shot distillable coherence agree in the first order term but \emph{disagree} in the second order term. In particular, the dependence on $\ve$ is qualitatively different in their upper and lower bounds. Thus, one could certainly argue that the bounds in~\cite{zhao2019one} are not as tight as they should be in the asymptotic limit.
\end{remark}



\section{Assisted coherence distillation and incoherent randomness extraction}\label{sec:assisted}

In this section we first introduce the resource theory of quantum coherence in the distributed scenario and then
formally define two different resource processing tasks with assistance from the environment within this resource
theory. The first is the assisted quantum coherence distillation originally investigated
in~\cite{chitambar2016assisted,streltsov2017towards}. The second is an variant of
randomness extraction task adapted into the distributed scenario. It plays a crucial role in establishing the second
order characterization of the former task.

\subsection{Resource theory of quantum coherence in distributed scenarios}\label{sec:distributed}

Motivated by the local operations and classical communication known from entanglement
theory~\cite{horodecki2009quantum}, Chitambar \textit{et al.}~\cite{chitambar2016assisted} seminally introduced the
framework of local incoherent operations and classical communication ($\LICC$) in the distributed resource theory of
quantum coherence which was later further explored by Streltsov \textit{et al.}~\cite{streltsov2017towards}. In this
scenario there are two separated parties, Alice and Bob, that are connected via a classical channel and restricted to
performing local incoherent operations. We think of Alice as assistant who helps Bob to manipulate coherence.
Here we briefly summarize several sets of free bipartite quantum operations widely studied in the distributed resource
theory of coherence and refer the interested readers to~\cite{streltsov2017towards,yamasaki2019hierarchy} for more
details:
\begin{itemize}
  \item $\mathbf{\LICC}$:
      the set of \emph{local incoherent operations and classical communication}~\cite{chitambar2016assisted}. That is,
      Alice and Bob perform local incoherent operations and share their outcomes via a classical channel. Throughout
      this work, we assume without loss of generality that the free local operations are chosen to be
      $\MIO$\footnote{One may choose other sets of local incoherent operations -- say $\SIO$ or $\IO$
      -- but all of these free local operations lead to the same result in the distributed setting~\cite[Section III]
         {regula2018nonasymptotic}. We choose $\MIO$ due to the fact that $\QIP$ reduces to $\MIO$ in the
         single partite scenario.};
  \item $\mathbf{\LQICC}$:
      the set of \emph{local quantum-incoherent operations and classical communication}~\cite{chitambar2016assisted}.
      That is, Alice can adopt arbitrary quantum operations while Bob is restricted to quantum incoherent operations,
      and they share the outcomes via a classical channel;
  \item $\mathbf{\SI}$: the set of \emph{separable incoherent operations}~\cite{streltsov2017towards}:
      \begin{align}\label{eq:sio}
          \Lambda_{AB\to A'B'}(\cdot) := \sum_i (A_i\ox B_i) (\cdot) (A_i\ox B_i)^\dagger,
      \end{align}
      where both $A_i$ and $B_i$ are incoherent Kraus operators satisfying
      $\sum_i A_i^\dagger A_i\ox B_i^\dagger B_i = \1_{AB}$;
  \item $\mathbf{\SQI}$: the set of \emph{separable quantum-incoherent operations}~\cite{streltsov2017towards}
      of the form~\eqref{eq:sio},
      where $B_i$ are incoherent Kraus operators satisfying
      $\sum_i A_i^\dagger A_i\ox B_i^\dagger B_i = \1_{AB}$.
\end{itemize}
The two free classes $\LQICC$ and $\SQI$ lead to the same set of free states, which is called the
\emph{quantum-incoherent states}~\cite{chitambar2016assisted,streltsov2017towards} (system $A$ is quantum and system
$B$ is incoherent) and bears the form
\begin{align}\label{eq:quantum-incoherent}
    \mathcal{QI}(A{:}B) := \left\{\sigma_{AB} = \sum_{b\in\cB}p_b\sigma^b_A\ox\proj{b}_B
                \sbar p_b\geq 0, \sum_{b\in\cB}p_b = 1,\; \sigma_A^b\in\cS(A)\right\}.
\end{align}
This motivates us to define the maximal set of
free operations that preserve $\mathcal{QI}$
-- the set of \emph{quantum-incoherent state preserving operations} $\mathbf{\QIP}$~\cite{yamasaki2019hierarchy}:
\begin{align}
  \QIP := \left\{\Lambda\in\channel{AB\to A^\prime B^\prime} \sbar
            \forall \sigma_{AB}\in\mathcal{QI},\;\Lambda_{AB\to A'B'}(\sigma_{AB})\in\mathcal{QI}\right\}.
\end{align}
By definition the following inclusion relations hold:
\begin{align}\label{eq:inclusion-relation}
    \LICC \subseteq \LQICC \subseteq \SQI \subseteq \QIP,\quad
    \LICC \subseteq \SI \subseteq \SQI\subseteq \QIP.
\end{align}

In what follows we assume $\cF\in\{\LICC,\LQICC,\SI,\SQI,\QIP\}$ be some 
chosen free \emph{bipartite} operation class, 
which is different from the free class $\cO\in\{\MIO,\DIO,\IO,\DIIO\}$ in
the previous section.

\subsection{Framework of assisted coherence distillation}\label{sec:assisted coherence distillation}

In the task of assisted coherence distillation, Alice and Bob work together to transform a given quantum state
$\rho_{AB}$ (not necessarily pure) to a MCS in system $B$ such that the error is within certain
threshold and the obtained MCS has rank as large as possible, under the constraint that the
available quantum operations are chosen from $\cF$.  We call this task the assisted coherence distillation since we can
think of Alice as a helpful environment who holds an \emph{extension} $\rho_{AB}$ of $\rho_B$ possessing certain amount
of quantum coherence. See Figure~\ref{fig:assisted-distillation} for illustration. More
formally, for any free operation class $\cF$, any given state $\rho_{AB} \in \cS(AB)$ and error tolerance $\ve \in
[0,1]$, the \emph{one-shot assisted distillable coherence of $\rho_{AB}$} is defined as
\begin{align}\label{eq:one-shot-EAD}
  C_{d,\cF}^{\ve}(\rho_{AB}) :=
    \max\left\{\log\vert B'\vert :
    P\left(\tr_{A'}\Lambda_{AB\to A'B'}(\rho_{AB}), \proj{\Psi_{B'}}\right)\leq\varepsilon,\Lambda\in\cF\right\},
\end{align}
where system $A'$ is at Alice's hand, system $B'$ is at Bob's hand, and $\ket{\Psi_{B'}}$ is a MCS in $B'$. 

\begin{figure}[!hbtp]
\centering
\includegraphics[width=0.6\textwidth]{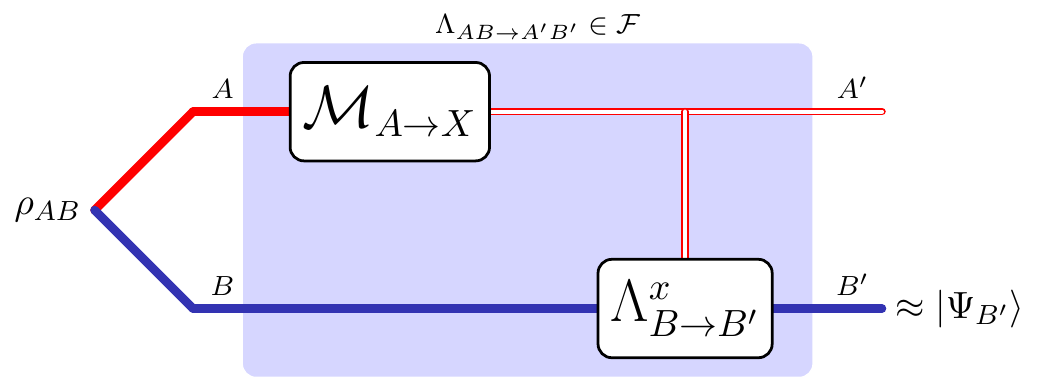}
\caption{Schematic diagram of the assisted coherence distillation. Alice and Bob together perform a free bipartite
        quantum operation $\cD_{AB\to A'B'}\in\cF$ to distill a MCS $\ket{\Psi_{B'}}$. The system
        in red belongs to Alice and the system in blue belongs to Bob. The shaded box depicts an one-way $\LQICC$
        strategy in which Alice performs a measurement $\cM_{A\to X}$ and sends the outcome $x$ to Bob. Conditioned on
        $x$, Bob performs an incoherent operation $\Lambda^x_{B\to B'}$ to distill $\ket{\Psi_{B'}}$.}
\label{fig:assisted-distillation}
\end{figure}

\subsection{Framework of assisted incoherent randomness extraction}\label{sec:assisted randomness extraction}

The task of assisted incoherent randomness extraction aims to obtain as many random bits as possible at Bob's
laboratory that is secure from possible attackers such as Eve, under the assistance of a helpful friend Alice. In the
beginning, Alice and Bob preshare a bipartite quantum state $\rho_{AB}$ with purification $\ket{\psi}_{RAB}$ such that
the reference system $R$ held by Eve. A general assisted incoherent randomness extraction protocol is characterized by
a triplet $(\Lambda,\Delta,f)$, where $\Lambda\in\cF$, and is composed of three steps:
\begin{enumerate}
\item Alice and Bob perform a free operation $\Lambda_{AB\to A'B'}\in\cF$ on their joint system. 
  Let $U_{AB\to A'B'E}$ be a Stinespring isometry representation of $\Lambda$. We
  assume the environment system $E$ of $\Lambda$ is also controlled by Eve. Since Alice is a friend of Bob, 
  Eve has no access to system $A'$. Hence Eve has control over two systems $ER$. 
  After the action of $\Lambda$, the whole system is in a pure state
  \begin{align}\label{eq:rho-Lambda-BER}
    \rho[\Lambda]_{A'B'ER} := U_{AB\to A'B'E}(\proj{\psi}_{RAB})U_{AB\to A'B'E}^\dagger.
  \end{align}
\item Bob dephases system $B'$ via the dephasing channel $\Delta_{B'}$.
  This yields the classical-quantum state
  \begin{align}
    \rho[\Lambda,\Delta]_{A'B'ER} := \Delta_{B'}(\rho[\Lambda]_{A'B'ER})
                                   = \sum_{b\in\cB}p_b\proj{b}_{C}\ox\sigma_{A'ER}^b,
  \end{align}
  where $p_b := \tr\bra{b}\rho[\Lambda]_{A'B'ER}\ket{b}$ 
  and $\sigma_{A'ER}^b:=\bra{b}\rho[\Lambda]_{A'B'ER}\ket{b}/p_b$.
\item A hash function $f$ is applied on $B'$ to extract the randomness that is \emph{secure from Eve},
  leading to the final output state
  \begin{align}
    \rho[\Lambda,\Delta,f]_{A'LER} := \rho[\Lambda,\Delta]_{A'f(B')ER}
    = \sum_{b\in\cB}p_b\proj{f(b)}_{L}\ox\sigma_{A'ER}^b.
  \end{align}
\end{enumerate}
A detailed procedure of the assisted randomness extraction via $(\Lambda,\Delta,f)$ is depicted in
Figure~\ref{fig:assisted-extraction}.

\begin{figure}[!hbtp]
\centering
\includegraphics[width=0.7\textwidth]{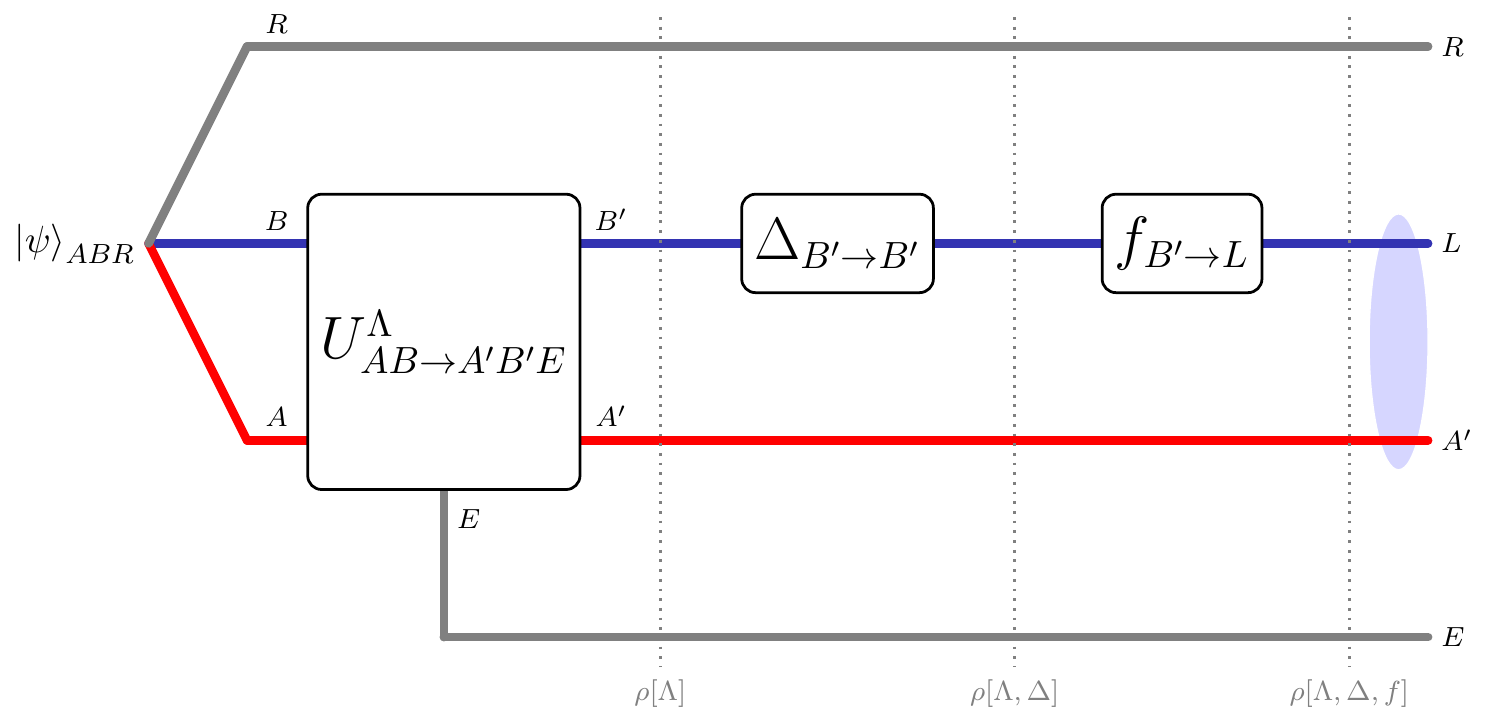}
\caption{Schematic diagram of an assisted incoherent randomness extraction protocol given by $[\Lambda,\Delta,f]$.
The system in red belongs to Alice, the system in blue belongs to Bob, and the systems in gray belong to Eve.}
\label{fig:assisted-extraction}
\end{figure}

Using the security metric~\eqref{min1}, the \emph{one-shot assisted extractable
randomness of $\rho_{AB}$} via $\Lambda$ is defined as
\begin{align}
    \ell_{\Lambda}^\varepsilon(\rho_{AB}) := \max_{f}\left\{\log |L| :
    d_{sec}\left(\rho[\Lambda,\Delta,f]_{LER}|ER\right)\leq \varepsilon \right\}.
\end{align}
Correspondingly, the \emph{one-shot assisted extractable randomness of $\rho_{AB}$} is defined as
\begin{align}\label{eq:one-shot-er}
    \ell_{\cF}^\varepsilon(\rho_{AB}) := \max_{\Lambda\in\cF}\ell_{\Lambda}^\varepsilon(\rho_{AB}).
\end{align}

Notice that the identity map $\id_{AB\to AB}$ is free in this distributed framework, 
thus $(\id_{AB\to AB},\Delta_B,f)$ is an valid assisted incoherent randomness extraction protocol for any hash function $f$. In this
protocol, Bob does not perform any assisted incoherent operation and directly dephases his state
$\psi_B\equiv\tr_{RA}\psi_{RAB}$ via $\Delta_B$. Hence, the environment system $\cH_E$
ceases to exist and we concern with the extractable randomness in state $\Delta_B(\psi_{RB})$ that is secure from $R$.
Following the argument around~\eqref{eq: hashing lemma}, we conclude for any $\eta\in (0,\varepsilon ]$ that
\begin{align}\label{eq: hashing lemma 2}
\ell_{\cF}^\varepsilon(\rho_{AB}) \geq
H_{\min}^{\varepsilon -\eta}(B\vert R)_{\sigma} + 4\log \eta-3,
\end{align}
where $\sigma_{RAB}:=\Delta_B(\psi_{RAB})$ is the dephased classical-quantum state and $H_{\min}^\ve$ is defined
in~\eqref{eq:conditional min}.

\begin{remark}
In the above assisted randomness extraction framework, we do not require that the extracted randomness by Bob is
secure from Alice. That is, it is possible that $L$ and $A'$ are classically correlated in the output state
$\rho[\Lambda,\Delta,f]_{A'LER}$ (cf. the shaded area of Figure~\ref{fig:assisted-extraction}). This assumption is
reasonable, since we regard Alice as a benevolent friend of Bob and aims to help him extracting randomness. We only
require that the extracted randomness is secure from the malicious Eve who has control over the systems $ER$.
\end{remark}

\begin{remark}\label{remark:randomness-extraction}
Assume that Bob ignores the assistance from Alice, i.e., Bob performs a free operation $\Lambda_{B\to B'}\in\DIIO$
(which is also free in the assisted framework) to extract randomness without touching $A$, we recover a slightly
more general randomness extraction framework than that was investigated in Section~\ref{sec:incoherent randomness
extraction} by allowing the state preshared between $B$ and the reference system $R$ to be a mixed state
$\psi_{RB}=\tr_A\psi_{RAB}$. Whereas in Section~\ref{sec:incoherent randomness extraction}, $\psi_{RB}$ is assumed to
be pure \textit{a priori}.
\end{remark}

\begin{remark}
Recently, Yang, Horodecki, and Winter initiated the research of distributed private randomness distillation~\cite{yang2019distributed}, in which Alice and Bob trust each
other and collaborate to extract independent randomness private against Eve.
Our assisted randomness extraction task is different from theirs mainly in two aspects:
First, we do not require that Alice and Bob share independent randomness -- all we need
is that Bob possesses randomness private against Eve; Second, our free operations
are naturally motivated within the resource theory of quantum coherence towards the distributed
scenario, which are different from their closed local operations and dephasing
channel communication. Actually, they mentioned in the Outlook section that the
distributed randomness extraction under the framework of coherence theory is a notable
topic to pursue. Our assisted randomness extraction task investigated here
is a first step towards this target.
\end{remark}

\subsection{Relation between assisted coherence distillation and assisted randomness extraction}

We present an equivalence between the assisted coherence distillation and the assisted incoherent randomness extraction
described above on the one-shot regime, extending Proposition~\ref{prop: IO protocol} to the assisted scenario.

\begin{theorem}[Exact one-shot connection]\label{thm:equivalence}
Let $\rho_{AB}$ be a bipartite quantum state and $\varepsilon\in[0,1]$. It holds that
\begin{align}
  C_{d,\QIP}^{\ve}(\rho_{AB})  = \ell_{\QIP}^\varepsilon(\rho_{AB}).\label{eq:equivalence}
\end{align}
\end{theorem}

The proof of Theorem~\ref{thm:equivalence} is divided into two pieces -- we first conclude the ``$\leq$'' direction in
Lemma~\ref{lemma:equivalence-converse} and then prove the ``$\geq$'' direction in
Lemma~\ref{lemma:equivalence-achievability-1}. Actually, for the ``$\leq$'' direction, we can obtain a much stronger
statement that $C_{d,\cF}^{\ve}(\rho_{AB})\leq\ell_{\cF}^\varepsilon(\rho_{AB})$ for arbitrary $\cF$. That is, the
optimal rate of assisted randomness extraction is always larger than the optimal rate of assisted coherence
distillation.

\begin{lemma}\label{lemma:equivalence-converse}
Let $\cF\in\{\LICC,\LQICC,\SI,\SQI,\QIP\}$. Let $\rho_{AB}$ be a bipartite quantum state and $\varepsilon\in[0,1]$. It
holds that
\begin{align}
  C_{d,\cF}^{\ve}(\rho_{AB}) \leq \ell_{\cF}^\varepsilon(\rho_{AB}).
\end{align}
\end{lemma}
\begin{proof}
This is shown by the same technique employed in proving Theorem~\ref{thm: one-shot relation}. Nevertheless, we write down the details for completeness.
Let $C_{d,\cF}^{\ve}(\rho_{AB}) = \log M$ and 
let $\Lambda_{AB\to A'B'}\in\cF$ achieves this rate, i.e.,
$P(\tr_{A'}\Lambda_{AB\to A'B'}(\rho_{AB}),\Psi_{B'})\leq\varepsilon$ and $\vert B'\vert = M$. Consider the following assisted
randomness extraction protocol $(\Lambda,\Delta,\id)$, where the hash function is chosen as the identity map $\id$.
Since $\rho[\Lambda]_{A'B'ER}$ purifies $\Lambda_{AB\to A'B'}(\rho_{AB})$ (it also purifies $\tr_{A'}\Lambda_{AB\to
A'B'}(\rho_{AB})$) and $\Psi_{B'}$ is a pure state, the Uhlmann's theorem~\cite{uhlmann1976transition} guarantees that
there exists a quantum state $\sigma^\ast_{A'ER}$ in $A'ER$ such that
\begin{align}\label{eq:jbxOB}
    P(\tr_{A'}\Lambda_{AB\to A'B'}(\rho_{AB}),\Psi_{B'}) 
  = P(\rho[\Lambda]_{A'B'ER}, \Psi_{B'}\ox\sigma^\ast_{A'ER}).
\end{align}
Consider the following chain of inequalities regarding state $\rho[\Lambda]_{LER}$:
\begin{align}
    d_{sec}(\rho[\Lambda,\Delta,\id]_{LER}|ER) 
:=&\;  \min_{\sigma_{ER} \in \cS(ER)}
        P(\rho[\Lambda,\Delta,\id]_{LER}, \pi_L \otimes \sigma_{ER}) \\
\stackrel{(a)}{\leq}&\;    P(\rho[\Lambda,\Delta,\id]_{LER}, \pi_L \otimes \sigma^\ast_{ER}) \\
\stackrel{(b)}{=}&\;  P\left( (\Delta_{B'}\ox\tr_{A'})(\rho[\Lambda]_{A'B'ER}), 
                              (\Delta_{B'}\ox\tr_{A'})(\Psi_{B'} \otimes \sigma^\ast_{A'ER})\right) \\
\stackrel{(c)}{\leq}&\;   
        P\left(\rho[\Lambda]_{A'B'ER}, \Psi_{B'} \otimes \sigma^\ast_{A'ER}\right) \\
\stackrel{(d)}{=}&\;     
        P(\tr_{A'}\Lambda_{AB\to A'B'}(\rho_{AB}),\Psi_{B'}) \\
\leq&\;   \varepsilon,
\end{align}
where $\sigma_{ER}^\ast:=\tr_{A'}\sigma_{A'ER}$ in $(a)$,
$(b)$ follows from $L\equiv B'$ and $\Delta_{B'}(\Psi_{B'})=\pi_{B'}$,
$(c)$ follows from the data-processing of purified distance, 
and $(d)$ follows from~\eqref{eq:jbxOB}. That is, $\log M$ is an achievable assisted randomness
extraction rate, yielding the desired inequality.
\end{proof}

\vspace*{\baselineskip}

The ``$\geq$'' direction in~\eqref{eq:equivalence} turns out to be more difficult. We first present a technical result
that is utilized in the proof of the this direction. This result is of independent interest as it offers a systematic
method to construct an assisted coherence distillation protocol from a given assisted randomness extraction protocol
(that does not use free operations) with the same performance.

\begin{proposition}\label{prop: QIP protocol}
Let $\id_{AB\to AB}$ is the identity map. Given an assisted randomness extraction protocol $(\id,\Delta,f)$ such that
$d_{sec}(\rho[\id,\Delta,f]_{LR}\vert R)\leq\varepsilon$, we can construct from the protocol a quantum operation
$\Gamma: AB \to L$ such that $\Gamma\in\QIP$ and $P(\Gamma_{AB\to L}(\rho_{AB}), \Psi_L) \leq
\varepsilon$, where $L$ is in Bob's possession.
\end{proposition}

The proof is given at the end of this section.

\begin{remark}
We emphasize that this technical result holds for the $\QIP$
class, i.e., we are only able to show that the devised quantum operation $\Gamma$ belongs to $\QIP$. This limits us to
enhance Theorem~\ref{thm:equivalence} (more concretely Lemma~\ref{lemma:equivalence-achievability-1} below) to other
free classes of operations. It remains as an interesting problem whether one can further construct a distillation
operation $\Gamma$ that belongs to a less powerful free class such as $\LICC$.
As a possible solution, we conceive an
alternative definition of the assisted incoherent randomness extraction task in Appendix~\ref{sec:alternative
formulation}. For this new variant, we are able to establish the equivalence relation~\eqref{eq:equivalence} for all free operation
classes under consideration. However, we have difficulty in obtaining one-shot achievability bound for that new task --
the lower bound concluded in~\eqref{eq: hashing lemma 2} does not hold any more. We left it as an open problem to
obtain complete characterization (one-shot, second order, and asymptotic analyses) for that task.
\end{remark}

\begin{lemma}\label{lemma:equivalence-achievability-1}
Let $\rho_{AB}$ be a bipartite quantum state and $\varepsilon\in[0,1]$. It holds that
\begin{align}
  C_{d,\QIP}^{\ve}(\rho_{AB}) \geq \ell_{\QIP}^\varepsilon(\rho_{AB}).
\end{align}
\end{lemma}
\begin{proof}
Let $\ell_{\QIP}^\varepsilon(\rho_{AB}) = \log\vert L\vert$ and let the randomness extraction protocol
$(\Lambda,\Delta,f)$ with $\Lambda_{AB\to A^{\prime}B^{\prime}}\in\QIP$ achieves this rate, i.e.,
$d_{sec}\left(\rho[\Lambda,\Delta,f]_{LER}|ER\right)\leq \varepsilon$. The key observation
is that applying the randomness extraction protocol $(\Lambda,\Delta,f)$ on quantum state $\rho_{AB}$ is the same as
applying the randomness extraction protocol $(\id,\Delta,f)$ on quantum state
$\sigma_{A^{\prime}B^{\prime}}\equiv\Lambda_{AB\to A^{\prime}B^{\prime}}(\rho_{AB})$ with purification
$\rho[\Lambda]_{A^{\prime}B^{\prime}ER}$ and reference systems $ER$. For the latter protocol, Proposition~\ref{prop:
QIP protocol} guarantees that there exists a quantum operation $\Gamma_{A^{\prime}B^{\prime}\to L}$
such that $\Gamma \in \QIP$ and $P(\Gamma(\sigma_{A^{\prime}B^{\prime}}),\Psi_L) \leq
\varepsilon$. Compositing these two quantum operations $\Lambda$ and $\Gamma$ yields
\begin{align}
  P(\Gamma_{A^{\prime}B^{\prime}\to L}
        \circ\Lambda_{AB\to A^{\prime}B^{\prime}}(\rho_{AB}),\Psi_L)
= P(\Gamma_{A^{\prime}B^{\prime}\to L}(\sigma_{A^{\prime}B^{\prime}}), \Psi_L)
\leq\varepsilon.
\end{align}
That is to say, the composite operation $\Gamma\circ\Lambda$ distills a MCS of rank $\vert L\vert$ such that the error
is bounded by $\varepsilon$. Since $\Lambda \in
\QIP$, $\Gamma \in \QIP$, and the class $\QIP$ is closed under composition~\cite{yamasaki2019hierarchy}, we have
$\Gamma\circ\Lambda\in\QIP$. As so, we have constructed an operation in $\QIP$ that distills
MCS at rate $\log\vert L\vert$ and thus complete the proof.
\end{proof}

\vspace*{\baselineskip}

Now we proceed to prove the method declared in 
Proposition~\ref{prop: QIP protocol}.
The following lemma is utilized, whose proof can be checked by definition.

\begin{lemma}\label{lem: fidelity}
    Let $\ket{u_i}$ and $\ket{v_i}$ be the purification of $\rho_i$ and $\sigma_i$
    such that $F(\rho_i,\sigma_i) = \<u_i|v_i\>$. Then we have
    \begin{align}\label{eq: app tmp1}
        F\left(\sum_i p_i \ket{i}\bra{i}\ox \rho_i, \sum_i q_i \ket{i}\bra{i} \ox \sigma_i\right)
    = \sum_i \sqrt{p_iq_i} F(\rho_i,\sigma_i)
    = F\left(\sum_i \sqrt{p_i}\ket{i}\ket{u_i},\sum_i \sqrt{q_i}\ket{i}\ket{v_i}\right).
    \end{align}
\end{lemma}

\noindent\textbf{[Proof of Proposition~\ref{prop: QIP protocol}]}
Let $\sigma^\ast_R$ be a quantum state that attains the minimum in
\begin{align}
  d_{sec}(\rho[\id,\Delta,f]_{LR}|R)
:=&\; \min_{\sigma_R \in \cS(R)} P(\rho[\id,\Delta,f]_{LR},\pi_L\otimes \sigma_R) \\
 =&\; P(\rho[\id,\Delta,f]_{LR},\pi_L\otimes \sigma^\ast_R) \leq\varepsilon.
\end{align}
The relation between purified distance $P$ and fidelity $F$ gives
\begin{align}\label{eq:EA-direct-tmp2}
    F(\rho[\id,\Delta,f]_{LR},\pi_L\otimes \sigma^*_{R}) \geq \sqrt{1-\varepsilon^2}.
\end{align}
We decompose the pure tripartite state $\ket{\psi}_{RAB}$ into the basis $\cB$ of system $B$ as
\begin{align}
  \ket{\psi}_{RAB} := \sum_{b\in\cB} \sqrt{p_B(b)} \ket{b}_B\vert\psi^b\rangle_{RA},
\end{align}
where $p_B$ a probability distribution and $\{\vert\psi^b\rangle\}$ a set of pure states in $RA$ which are not
necessarily mutually orthogonal. Define the incoherent isometry $U_f$ from $\cH_B$ to $\cH_{LB}$ such that
\begin{align}
\forall b\in\cB,\; U_f \ket{b}_B:= \ket{f(b)}_L\ket{b}_B.
\end{align}
Notice that
\begin{align}
    U_f\ket{\psi}_{ABR}
=  \sum_{b\in\cB} \sqrt{p_B(b)} U_f\ket{b}_B\vert\psi^b\rangle_{RA}
=  \sum_{b\in\cB} \sqrt{p_B(b)} \ket{f(b)}_L\ket{b}_B\vert\psi^b\rangle_{RA}.
\end{align}
Since the mapping $f:\cB\to \cL$ is non-injective, where $\cL$ is the alphabet of $L$, each new basis $\ell$ may
 corresponds to several bases $b$ such that $f(b)=\ell$. As so, we can equivalently write $U_f\ket{\psi}_{ABR}$ as
\begin{align}
U_f\ket{\psi}_{ABR} = \sum_{\ell\in \cL} \sqrt{r_\ell} \ket{\ell}_L \otimes \vert\phi^\ell\rangle_{RAB},
\end{align}
where the normalized vectors $\vert\phi^\ell\rangle_{RAB}$ and normalization factors $r_\ell$ for $\ell\in\cL$ satisfy
\begin{align}
        \sqrt{r_\ell} \vert\phi^\ell\rangle_{RAB}
:= \sum_{b \in \cB: f(b) = \ell} \sqrt{p_B(b)}\ket{b}_B\otimes \vert\psi^b\rangle_{RA}.
\end{align}

Let $\ket{\phi^*}_{RAB}$ on $\cH_{RAB}$ be a purification of $\sigma^\ast_R$. By the Uhlmann's
theorem~\cite{uhlmann1976transition}, for each $\vert\phi^\ell\rangle_{RAB}$ there exists a unitary $U_\ell$ on
$\cH_{AB}$ such that
\begin{align}\label{eq:EA-direct-tmp3}
F(\tr_{AB} \proj{\phi_\ell}_{RAB},\sigma^*_{R})
= F(U_\ell\vert\phi^\ell\rangle_{RAB},\ket{\phi^*}_{RAB}).
\end{align}
Define the conditional unitary $U := \sum_{\ell \in \cL} \proj{\ell}_L \otimes U_\ell$.
We have
\begin{align}
&\; F\left(UU_f \ket{\psi}_{RAB},\; \ket{\Psi}_L \otimes \ket{\phi^*}_{RAB}\right)  \\
\stackrel{(a)}{=}&\;
    F\left(\sum_{\ell\in \cL} \sqrt{r_\ell} \ket{\ell}_L \otimes U_{\ell}\vert\phi^\ell\rangle_{RAB},\;
          \sum_{\ell\in\cL} \frac{1}{\sqrt{\vert\cL\vert}}\ket{\ell}_L \otimes \ket{\phi^*}_{RAB}\right)\\
\stackrel{(b)}{=}&\;
    \sum_{\ell\in \cL}\sqrt{\frac{r_\ell}{\vert\cL\vert}}
    F\left(U_\ell \vert\phi^\ell\rangle_{RAB}, \ket{\phi^*}_{RAB}\right)\\
\stackrel{(c)}{=}&\;
    \sum_{\ell\in \cL}\sqrt{\frac{r_\ell}{\vert\cL\vert}}
    F\left(\tr_{AB} \proj{\phi_\ell}_{RAB},\sigma^*_{R}\right) \\
\stackrel{(d)}{=}&\;
    F\left( \sum_{\ell\in\cL} r_\ell \proj{\ell}_L\otimes \tr_{AB}\proj{\phi^\ell}_{RAB},
            \sum_{\ell\in\cL}\frac{1}{\vert\cL\vert}\proj{\ell}_L\otimes\sigma_R^*\right)\\
\stackrel{(e)}{=}&\;
    F\left(\rho[\id,\Delta,f]_{LR}, \pi_L\otimes \sigma^*_{R}\right), \label{eq:EA-direct-tmp1}
\end{align}
where $(a)$ follows by definition, $(b)$ and $(d)$ follow from Lemma~\ref{lem: fidelity}, $(c)$ follows from
Eq.~\eqref{eq:EA-direct-tmp3}, and $(e)$ follows from the fact that
\begin{align}
  \rho[\id,\Delta,f]_{LR}
= \tr_{AB} U_f\Delta_B\left(\Psi_{RAB}\right)U_f^\dagger
= \sum_{\ell\in \cL}r_\ell\proj{\ell}_L\otimes \tr_{AB} \proj{\phi^\ell}_{RAB}.
\end{align}

Now we construct the required quantum operation $\Gamma_{AB\to L}$ as
\begin{align}
  \Gamma_{AB\to L}(\cdot) := \tr_{AB}\left[UU_f (\cdot) U_f^\dagger U^\dagger\right].
\end{align}
See Figure~\ref{fig:distillation-protocol} for illustration of this construction. First we show that $\Gamma_{AB\to L}$
can transform $\rho_{AB}$ to a MCS $\Psi_L$ within error $\varepsilon$. Consider the following chain of inequalities:
\begin{align}
    F\left(\Gamma_{AB\to L}(\rho_{AB}), \proj{\Psi}_L\right)
&=  F\left(\tr_{AB}\left[UU_f\rho_{AB}U_f^\dagger U^\dagger\right],
            \proj{\Phi}_L\right) \\
&\stackrel{(a)}{\geq} F\left(UU_f \ket{\psi}_{RAB}, \ket{\Phi}_L\ox\ket{\phi^*}_{RAB}\right) \\
&\stackrel{(b)}{=} F\left(\rho[\id,\Delta,f]_{LR}, \pi_L\otimes \sigma^*_{R}\right) \\
&\stackrel{(c)}{\geq} \sqrt{1-\varepsilon^2},
\end{align}
where $(a)$ follows by the data-processing inequality of the fidelity $F$ under $\tr_{ABR}$, $(b)$ follows
by~\eqref{eq:EA-direct-tmp1}, and $(c)$ follows from~\eqref{eq:EA-direct-tmp2}. This implies that
$P(\Gamma_{AB\to L}(\rho_ {AB}), \Psi_L) \leq \varepsilon$.
It remains to check that $\Gamma_{AB\to L}\in\QIP$. Note that $\Gamma$ admits the following Kraus decomposition
\begin{align}
    \Gamma(\cdot) = \sum_ {b\in \cB} K_b U_f (\cdot) U_f^\dagger K_b^\dagger
\end{align}
with operators $K_b:=\bra{b}U$. This is indeed a Kraus decomposition since
\begin{align}
  \sum_ {b\in \cB}\left(U_f^\dagger K_b^\dagger\right) \left(K_b U_f\right)
= U_f^\dagger\left(\sum_{b\in \cB}K_b^\dagger K_b\right)U_f
= U_f^\dagger U^\dagger \left(\sum_{b\in\cB}\proj{b}_B\right)U U_f
= U_f^\dagger U^\dagger U U_f = \1.
\end{align}
On the other hand, one can verify that for any bases $b,b^\prime \in \cB$ and any quantum state $\rho_A^{b^\prime}$ in
$A$,
\begin{align}
  K_b U_f\left(\proj{b^\prime}_B\otimes\rho_A^{b^\prime}\right) U_f^\dagger K_b^\dagger
= \tr\left[K_b \rho_A^{b^\prime}  K_b^\dagger\right]
  \vert\langle b \vert U_{f(b^\prime)} \vert b^\prime\rangle\vert^2
  \proj{f(b^\prime)}_B \in \cI^{**}(B).
\end{align}
That is to say, $\Gamma$ transforms arbitrary quantum-incoherent state in $\mathcal{QI}$ to some incoherent state in
$\cI$, which in turn is a strict subset of $\mathcal{QI}$. Thus we can conclude that $\Gamma \in \QIP$.
We are done.
\hfill $\blacksquare$

\begin{figure}[!hbtp]
\centering
\includegraphics[width=0.6\textwidth]{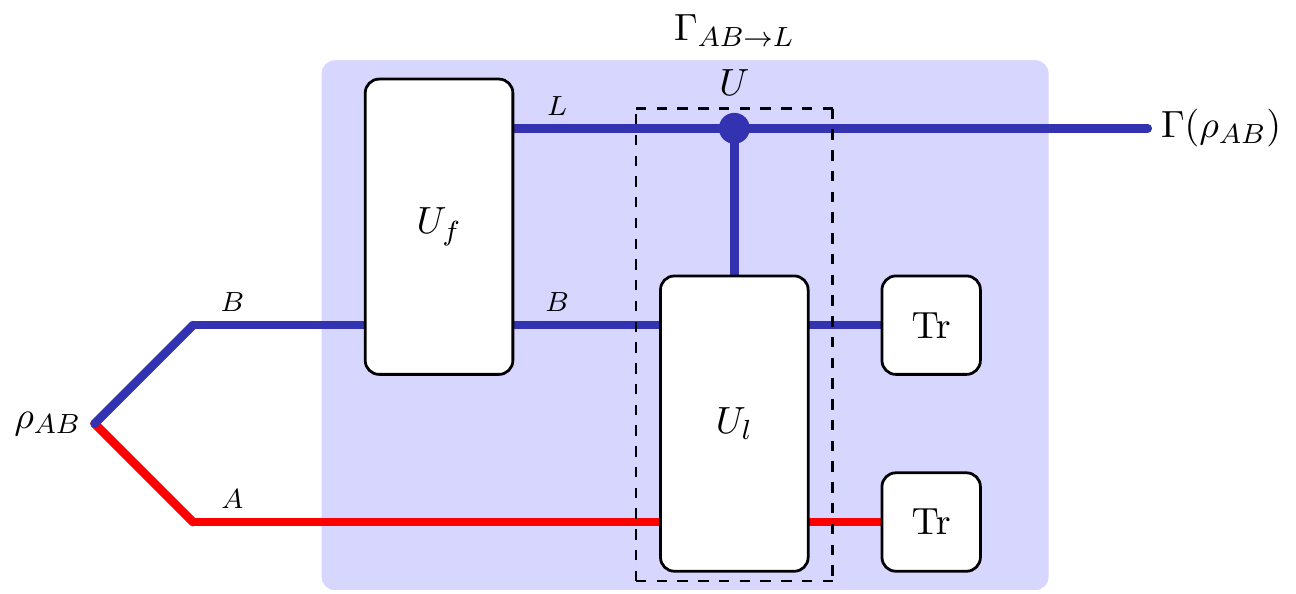}
\caption{Schematic diagram of the $\QIP$ distillation protocol in 
        Proposition~\ref{prop: QIP protocol}.
        The system in red belongs to Alice and the system in blue belongs to Bob.
        The isometry $U_f$ and conditional unitary $U$ in the dashed box are not 
        necessarily incoherent.
        The shaded area depicts the designed operation $\Gamma_{AB\to L}$, 
which belongs to $\QIP$.}
\label{fig:distillation-protocol}
\end{figure}

\subsection{Second order analysis}

We now give second order characterization to both the assisted coherence distillation and the assisted randomness
extraction. We do so by first establishing one-shot achievability and converse bounds for these two tasks in terms of
the hypothesis testing relative entropy. These two bounds have matching dependence on the error threshold
$\varepsilon$. Then we derive a second order characterization by invoking the second order expansion of the hypothesis
testing relative entropy. As a direct corollary, we show that the ultimate rate in the asymptotic setting is uniquely
determined by the quantum incoherent relative entropy of coherence.

\begin{theorem}[One-shot characterization]\label{thm:one-shot-characterization}
Let $\cF\in\{\LICC,\LQICC,\SI,\SQI,\QIP\}$.
Let $\rho_{AB}$ be a bipartite quantum state and $\varepsilon\in(0,1)$.
For arbitrary $\eta\in(0,\varepsilon)$ and $\delta\in(0,\min\{(\varepsilon-\eta)^2/3, 1-(\varepsilon-\eta)^2\})$,
it holds that
\begin{subequations}\label{eq:one-shot characterization}
\begin{align}
D_H^{(\varepsilon-\eta)^2-2\delta}\left(\rho_{AB}\rel\Delta_B(\rho_{AB})\right)
- c(\rho_{AB},\varepsilon,\delta,\eta)
&\stackrel{\text{\ding{172}}}{\leq} \ell_{\cF}^\varepsilon(\rho_{AB}) \label{eq:one-shot characterization-1} \\
&\stackrel{\text{\ding{173}}}{\leq} \ell_{\QIP}^\varepsilon(\rho_{AB}) \\
&\stackrel{\text{\ding{174}}}{=}  C_{d,\QIP}^{\ve}(\rho_{AB})
\stackrel{\text{\ding{175}}}{\leq}
D_H^{\varepsilon^2}\left(\rho_{AB}\rel\Delta_B(\rho_{AB})\right),\label{eq:one-shot characterization-2}
\end{align}
\end{subequations}
where the correction term $c$ is defined as
\begin{align}
  c(\rho_{AB},\varepsilon,\delta,\eta)
:=&\; \log \theta(\rho_{AB}) + \log \theta(\Delta_B(\rho_{AB}))\nonumber \\
  &\quad + \log (({\varepsilon }-\eta)^2-\delta)-\log (\delta^5\eta^4 ({\varepsilon }-\eta)^2
  (1-({\varepsilon }-\eta)^2+\delta)) + 11.\label{eq:correction term}
\end{align}
\end{theorem}

The proof of Theorem~\ref{thm:one-shot-characterization} is divided into the following pieces: the achievability part
\text{\ding{172}} is drawn in Lemma~\ref{lemma:achievability}; the inequality \text{\ding{173}} follows directly from
the inclusion relation~\eqref{eq:inclusion-relation}; the equality \text{\ding{174}} has already been shown in
Theorem~\ref{thm:equivalence}; and the converser part \text{\ding{175}} is proved in Lemma~\ref{lemma:converse}.
However, we note that the one-shot characterization for the assisted coherence distillation
$C_{d,\QIP}^{\ve}(\rho_{AB})$ cannot be generalized to other classes of free operations due to that
$\text{\ding{174}}$ holds only for $\QIP$.

\begin{lemma}[Achievability]\label{lemma:achievability}
Let $\cF\in\{\LICC,\LQICC,\SI,\SQI,\QIP\}$.
Let $\rho_{AB}$ be a bipartite quantum state and $\varepsilon\in(0,1)$.
For arbitrary $\eta\in(0,\varepsilon)$ and $\delta\in(0,\min\{(\varepsilon-\eta)^2/3, 1-(\varepsilon-\eta)^2\})$,
it holds that
\begin{align}\label{eq:achievability}
\ell_{\cF}^\varepsilon(\rho_{AB}) \geq
D_H^{(\varepsilon-\eta)^2-2\delta}\left(\rho_{AB}\rel\Delta_B(\rho_{AB})\right)
- c(\rho_{AB},\varepsilon,\delta,\eta).
\end{align}
where the correction term $c$ is defined in~\eqref{eq:correction term}.
\end{lemma}
\begin{proof}
In~\eqref{eq: hashing lemma 2} we already obtained the following lower bound for arbitrary free class $\cF$
in terms of the smooth conditional min-entropy:
\begin{align}\label{eq:hashing lemma QIP}
\ell_{\cF}^\varepsilon(\rho_{AB})
\geq H_{\min}^{\varepsilon -\eta}(B|R)_{\sigma} + 4\log \eta-3,
\end{align}
where $\sigma_{RAB}:=\Delta_B(\proj{\psi}_{RAB})$ is the dephased classical-quantum state. Moreover, we establish
in Proposition~\ref{prop:relations} a relation between the smooth min entropy $H_{\min}^\varepsilon$ and the hypothesis
testing relative entropy $D_H^\varepsilon$ regarding state $\sigma_{RAB}$.
Adapting~\eqref{eq: Hmin DH} into~\eqref{eq:hashing lemma QIP}, we reach the desired~\eqref{eq:achievability}.
\end{proof}

\vspace*{0.1in}

Our converse bound proved below asserts that the quantum-incoherent hypothesis testing relative entropy of coherence
upper bounds the one-shot assisted distillable coherence of $\rho_{AB}$ via arbitrary free operations
$\cF\in\{\LICC,\LQICC,\SI,\SQI,\QIP\}$. This result can be viewed as an one-shot analog of~\cite[Theorem
2]{streltsov2017towards}, in which Streltsov \textit{et al.} showed that the quantum-incoherent relative entropy of
coherence upper bounds the asymptotic assisted distillable coherence of $\rho_{AB}$. Note that \text{\ding{175}} in~\eqref{eq:one-shot
characterization-2} is immediately concluded from Lemma~\ref{lemma:converse} by choosing
$\sigma_{AB}\equiv\Delta_B(\rho_{AB})$ in~\eqref{eq:converse}.

\begin{lemma}[Converse]\label{lemma:converse}
Let $\cF\in\{\LICC,\LQICC,\SI,\SQI,\QIP\}$.
Let $\rho_{AB}$ be a bipartite quantum state and $\varepsilon\in(0,1)$. It holds that
\begin{align}\label{eq:converse}
  C^\varepsilon_{d,\cF}(\rho_{AB})
\leq \min_{\sigma_{AB}\in\mathcal{QI}}D_H^{\varepsilon^2}\left(\rho_{AB}\rel\sigma_{AB}\right).
\end{align}
\end{lemma}

\begin{proof}
Assume $C^\varepsilon_{d,\cF}(\rho_{AB}) = \log M$, that is, there exists a free operation $\cD_{AB\to C}\in\cF$
such that
\begin{align}
    P\left(\cD_{AB\to C}(\rho_{AB}), \Psi_M\right)\leq\varepsilon.
\end{align}
By the definition of purified distance $P$, the above condition is equivalent to
\begin{align}\label{eq:converse-tmp1}
 \tr \left[\cD_{AB\to C}(\rho_{AB}) \Psi_M\right]  \geq 1-\varepsilon^2.
\end{align}

Set $M_{AB}:=\cD^\dagger(\Psi_M)$, where $\cD^\dagger$ is the adjoint of $\cD$.
Since $\cD$ is completely positive, so is $\cD^\dagger$. This gives $M_{AB}\geq 0$.
On the other hand, condition~\eqref{eq:converse-tmp1} guarantees that $M_{AB}\leq \1_{AB}$.
Now we have
\begin{align}
    \tr\left[\rho_{AB}M_{AB}\right]
=  \tr\left[\rho_{AB}\cD^\dagger(\Psi_M)\right]
= \tr\left[\cD(\rho_{AB})\Psi_M\right] \geq 1-\varepsilon^2.
\end{align}
That is to say, $M_{AB}$ is a feasible solution for $D_H^{\varepsilon^2}\left(\rho_{AB}\rel\sigma_{AB}\right)$. Then
\begin{align}
      D_H^{\varepsilon^2}\left(\rho_{AB}\rel\sigma_{AB}\right)
&\geq  -\log \tr\left[\sigma_{AB}M_{AB}\right] \\
&=   -\log \tr\left[\sigma_{AB}\cD^\dagger(\Psi_M)\right] \\
&=   -\log \tr\left[\cD(\sigma_{AB})\Psi_M\right] \\
&\geq -\log\frac{1}{M} \\
&= C^\varepsilon_{d,\cF}(\rho_{AB}),
\end{align}
where the second inequality follows from that $\cD$ is quantum-incoherent state preserving and thus $\cD(\sigma_{AB})$
is a incoherent state in $C$. This, together with Lemma~\ref{lemma:incoherent-bound}, leads to the desired
inequality. We are done.
\end{proof}

\begin{lemma}\label{lemma:incoherent-bound}
Let $\sigma\in\cI$ be an incoherent quantum state.
It holds that $\langle\Psi_M\vert\sigma\vert\Psi_M\rangle\leq 1/M$.
\end{lemma}

The matching dependence on the error parameter $\varepsilon$ of the achievable and converse bounds in
Theorem~\ref{thm:one-shot-characterization} yields a second order expansion of the assisted distillable
coherence and also the assisted incoherent extractable randomness. The result is useful in refining the
optimal rate and determining the convergence rate of the assisted distillable coherence to its first order coefficient.

\begin{theorem}[Second order expansion]\label{thm: second order EA}
Let $\cF\in\{\LICC,\LQICC,\SI,\SQI,\QIP\}$.
Let $\rho_{AB}$ be a bipartite quantum state and $\varepsilon\in[0,1]$.
The following second order expansions hold:
\begin{subequations}
\begin{align}
    C_{d,\QIP}^{\ve}\left(\rho_{AB}^{\ox n}\right) & = nD\left(\rho_{AB}\rel\Delta_B(\rho_{AB})\right)
  + \sqrt{nV\left(\rho_{AB}\rel\Delta_B(\rho_{AB})\right)}\Phi^{-1}(\varepsilon^2) + O(\log n),\\
  \ell_{\cF}^\varepsilon\left(\rho_{AB}^{\ox n}\right)
& = nD\left(\rho_{AB}\rel\Delta_B(\rho_{AB})\right)
  + \sqrt{nV\left(\rho_{AB}\rel\Delta_B(\rho_{AB})\right)}\Phi^{-1}(\varepsilon^2) + O(\log n).
\end{align}
\end{subequations}
\end{theorem}

\begin{proof}
The proof follows from Theorem~\ref{thm:one-shot-characterization}
and a similar argument of the proof of Theorem~\ref{thm: second order}.
\end{proof}

\begin{remark}
From Proposition~\ref{prop:relations} in Section~\ref{sec: tripartite state}, the first and second order asymptotics can also be written as conditional entropy $H(B|R)_{\sigma_{ABR}}$ and conditional information variance $V(B|R)_{\sigma_{ABR}}$ respectively, where $\sigma_{ABR}= \rho[\id,\Delta]$ is the dephased classical-quantum state from the assisted randomness extraction protocol.
\end{remark}

\begin{remark}\label{remark:asymptotic}
As a consequence of Theorem~\ref{thm: second order EA}, we find the quantum-incoherent relative entropy of
coherence, defined as $C^{A\vert B}_R(\rho_{AB})
:=\min_{\sigma_{AB}\in\mathcal{QI}}D\left(\rho_{AB}\rel\sigma_{AB}\right)
=D\left(\rho_{AB}\rel\Delta_B(\rho_{AB})\right)$~\cite[Eq. (4)]{chitambar2016assisted}, quantifies the ultimate power
of assistance in both the coherence distillation and incoherent randomness extraction tasks, in the sense that it is
the best distillation and extraction rates that can be achieved using the largest free class $\QIP$ in the distributed
scenario. This empowers the quantum-incoherent relative entropy of coherence measure a new operational meaning.
\end{remark}

\begin{remark}
Our results -- the single-shot in Theorem~\ref{thm:one-shot-characterization}, the second order in
Theorem~\ref{thm: second order EA}, and its corollary in Remark~\ref{remark:asymptotic} -- together paint an
`almost' complete picture for the assisted coherence distillation task. What's more, they cover many known results as
special cases:
\begin{enumerate}
  \item When $\rho_{AB}=\rho_A\ox\rho_B$ is a product state, our results reduce to the single partite coherence
      distillation~\cite{winter_2016}. This is so because $\QIP$ reduces to $\MIO$ in single party setting.
      Theorem~\ref{thm: second order EA} matches the second order
      results for the coherence distillation without assistance in Theorem~\ref{thm: second order}.
  \item When $\rho_{AB}$ is pure or maximally correlated, Remark~\ref{remark:asymptotic} enhances the results of
      Theorem 5 and Proposition 6 in~\cite{streltsov2017towards} by stating that the quantum-incoherent relative
      entropy of coherence $C^{A\vert B}_R(\rho_{AB})$ is the ultimate rate that can be achieved in the assisted
      coherence distillation even if we make use of the largest free operation class $\QIP$. Note that the same
      conclusion was previously obtained in~\cite[Proposition 19]{yamasaki2019hierarchy} for the pure state case.
\end{enumerate}
\end{remark}

\section{Strong converse property}
\label{sec: strong converse}

The direct part of resource distillation states that for any rate below the optimal rate, there is a corresponding distillation protocol that accomplishes the task successfully. More precisely, if we denote the transformation error in the protocol for $n$ uses of the underlying resource by $\ve$, then for any rate below the optimal rate there exists a protocol, whose transformation error $\ve$ vanishes in the asymptotic limit $n\to +\infty$. Such rates are called achievable, and the optimal rate is defined as the supremum over all achievable rates. In contrast, the converse part states that for any distillation protocol with a rate above the optimal rate, the error does not vanish asymptotically, that is, it is bounded away from 0 in the asymptotic limit $n\to +\infty$. This is usually called weak converse. In principle, it leaves open the possibility of a trade-off between error and rate of a protocol. However, the strong converse property rules out such a possibility, stating that for any distillation protocol with a rate above the optimal rate, the corresponding transformation error $\ve$ incurred in the protocol converges to one. In other words, such protocols become worse with increasing block length $n$, and eventually fail with certainty in the asymptotic limit. In this part we showcase a standard argument how a second order result automatically implies the
strong converse property.

As a concrete example, we consider the unassisted coherence distillation whose strong converse property has been pointed out by~\cite[Theorem 16]{zhao2019one}. Here we give an alternative proof. For simplicity, we denote $C_r(\rho) := D(\rho\|\Delta(\rho))$ and $V_r(\rho) :=
V(\rho\|\Delta(\rho))$. For any achievable rate $R_n$, we have $R_n \leq \frac{1}{n}C_{d,\cO}^{\ve}(\rho^{\ox n})$. By
Theorem~\ref{thm: second order}, we have
\begin{align}\label{eq: second order rate}
R_{n} \leq C_r(\rho) + \sqrt{\frac{V_r(\rho)}{n}}\, \Phi^{-1}(\ve^2) + f(n) \quad \text{with} \quad f(n) \in O\left(\frac{\log n}{n}\right).
\end{align}
Rearranging~\eqref{eq: second order rate} and using monotonicity of $\Phi$ yields
\begin{align}\label{eq: strong converse}
    \ve^2 \geq  {\Phi\left(\sqrt{\frac{n}{V_r(\rho)}} (R_n - C_r(\rho)) + g(n)\right)} \quad \text{with} \quad g(n) = -\frac{\sqrt{n}f(n)}{\sqrt{V_r(\rho)}}.
\end{align}
Thus $\lim_{n\to +\infty} g(n) = 0$. Note that $\lim_{x\to +\infty} \Phi(x) = 1$. For any achievable rate $R_n > C_r(\rho)$, the argument in~\eqref{eq: strong converse} diverges to $+\infty$ and thus we have $\ve \to 1$ as $n \to \infty$. This implies the strong converse property of coherence distillation under $\cO \in \{\MIO,\DIO,\IO,\DIIO\}$. Similar argument works for the incoherent randomness extraction.

Moreover, following the same argument outlined above, we can conclude from Theorem~\ref{thm: second order EA}
that both the assisted coherence distillation via $\QIP$ and the assisted incoherent randomness extraction
via arbitrary free operation class $\cF\in\{\LICC,\LQICC,\SI,\SQI,\QIP\}$ satisfy the strong converse property. This
promises the unique role of $C^{A\vert B}_R(\rho_{AB})$ in the two assisted tasks.

\section{Relations among entropies of a dephased tripartite quantum state}\label{sec: tripartite state}

Let $\rho_{AB}$ be a bipartite quantum state with purification $\ket{\psi}_{RAB}$, i.e., $\tr_R\psi_{RAB}=\rho_{AB}$.
Assume the following decomposition into the basis $\cB$ of system $B$:
\begin{align}\label{eq:purification}
  \ket{\psi}_{RAB} := \sum_{b\in\cB} \sqrt{p_B(b)} \ket{b}_B\vert\psi^b\rangle_{RA},
\end{align}
where $p_B$ a probability distribution and $\{\vert\psi^b\rangle\}$ a set of pure states in $AR$ which are not
necessarily mutually orthogonal. By definition, we have
$\rho_{AB} = \tr_R\psi_{RAB}$. Dephasing $B$ yields the classical-quantum state
\begin{align}\label{eq:dephasing}
  \sigma_{RAB} := \Delta_B(\psi_{RAB})
                   = \sum_{a\in\cB}p_b\proj{b}_B\ox\proj{\psi^b}_{RA}.
\end{align}
Notice that $\sigma_{AB} = \Delta_B(\rho_{AB})$. In the following, we establish a list of relations among the various
entropies evaluated on the dephased quantum state $\sigma_{RAB}$. These relations are essential in the above second
order analysis. We regard these relations are of independent interests and may find applications in other quantum
information processing tasks.

Before presenting the relations, we first introduce some notations. For two Hermitian operators $X$ and $Y$, we denote
by $\{X\geq Y\}$ the projector onto the space spanned by the eigenvectors of $X-Y$ with non-negative eigenvalues. Let
$\rho\in\cS(\cH)$, $\sigma\in\cP(\cH)$, and $\varepsilon\in[0,1]$. The \emph{information spectrum relative
entropy} of $\rho$ w.r.t. $\sigma$ is defined as~\cite[Definition 8]{tomamichel2013hierarchy}
\begin{align}
    D_s^\varepsilon (\rho\|\sigma) := \sup \left\{x \sbar \tr \rho \{\rho \leq 2^x \sigma\} \leq \varepsilon \right\}.
\end{align}
The Nussbaum and Szko\l{}a's distributions are used intensively in the analysis that follows.
Assume the eigenvalue decompositions $\rho=\sum_x r_x
\ketbra{v_x}{v_x}$ and $\sigma=\sum_y s_y \ketbra{u_y}{u_y}$. Their Nussbaum and Szko\l{}a's
distributions~\cite{nussbaum2009} are defined as $P_{\rho,\sigma}(x,y):=r_x|\langle v_x|u_y\rangle |^2$ and
$Q_{\rho,\sigma}(x,y):=s_y|\langle v_x|u_y\rangle |^2$. These distributions satisfy the property that
\begin{align}\label{eq: PQ first and second order}
    D(P_{\rho,\sigma}\|Q_{\rho,\sigma}) = D(\rho\|\sigma) \quad \text{and} \quad V(P_{\rho,\sigma}\|Q_{\rho,\sigma}) = V(\rho\|\sigma).
\end{align}

\begin{proposition}\label{prop:relations}
Let $\varepsilon\in(0,1)$ and $\delta\in(0,\min\{\varepsilon ^2/3,1-\varepsilon ^2\})$. We have the following
relations regarding the dephased tripartite quantum state $\sigma_{RAB}$~\eqref{eq:dephasing}:
\begin{subequations}
\begin{align}
D_s^\varepsilon(P_{\sigma_{BR},\1_B \ox \sigma_R}
\|Q_{\sigma_{BR},\1_B \ox \sigma_R})
&= - D_s^{1-\varepsilon }(P_{\rho_{AB},\Delta_B(\rho_{AB})}
\|Q_{\rho_{AB},\Delta_B(\rho_{AB})}),\label{eq: reduction connection 1}\\
D(\sigma_{BR}\|\1_B \ox \sigma_R)
& = - D(\rho_{AB}\|\Delta_B(\rho_{AB})),\label{eq: reduction connection 2}\\
V(\sigma_{BR}\|\1_B \ox \sigma_R) & = V(\rho_{AB}\|\Delta_B(\rho_{AB})),\label{eq: reduction connection 3} \\
H_{\min}^{\varepsilon }(B|R)_{\sigma}
&\geq D_H^{\varepsilon ^2-2\delta}(\rho_{AB}\|\Delta_B(\rho_{AB})) - c(\rho_{AB},\varepsilon ,\delta),
\label{eq: Hmin DH}
\end{align}
\end{subequations}
where the correction term is defined as
\begin{align}\label{eq:correction term 2}
    c(\rho_{AB},\varepsilon ,\delta)
:=  \log \theta(\rho_{AB})
  + \log \theta(\Delta_B(\rho_{AB}))
  + \log (\varepsilon^2-\delta)-\log (\delta^5\varepsilon^2(1-\varepsilon ^2+\delta)) + 8.
\end{align}
\end{proposition}

\noindent\textbf{[Proof of Eq.~\eqref{eq: reduction connection 1}]}
For reduced state $\sigma_R$ assume the spectral decompositions $\sigma_R=\sum_rq_r \proj{v^r_R}$, For the
reduced states $\{\psi_R^b\}_b$, assume the spectral decompositions
\begin{align}
    \psi_R^b = \sum_rp_{R\vert B}(r\vert b)\proj{v^{r\vert b}}_R,
\end{align}
where $p_{R\vert B}$ is a conditional probability distribution, and
$\{\vert v^{r\vert b}\rangle\}$ is an
orthonormal basis of $R$ for each $b$. Set
$p_{BR}:=p_{R\vert B}p_B$, we have the following spectral decompositions
\begin{align}
  \1_B \ox \sigma_R &= \sum_{b,r} q_r \proj{b}_B\ox\proj{v^r}_R, \\
  \sigma_{BR} &= \sum_{b,r} p_{BR}(b,r)\proj{b}_B\ox\proj{v^{r\vert b}}_R.
\end{align}
Notice that $\sigma_R$ and $\rho_{AB}$ are two marginal states of pure tripartite state $\Psi_{ABR}$ and thus they have
the same eigenvalues due to the Schmidt theorem. That is, we can denote the eigenvalue decomposition of $\rho_{AB}$ as
$\rho_{AB} = \sum_r q_r \proj{u^r}$, where $\{\vert u^r\rangle\}$ is an orthonormal basis of $AB$. What's more, since
each conditional sate $\vert\psi^b\rangle_{RA}$~\eqref{eq:dephasing} is pure, the two marginal states $\psi^b_R$ and
$\psi^b_A$ has the same eigenvalues from the Schmidt theorem.
That is,
there exists $\{\vert u^{a\vert b}\rangle\}$ is an orthonormal basis of $A$ for each $b$ such that
\begin{align}\label{eq:Schmidt-decomposition}
  \vert\psi^b\rangle_{RA}
= \sum_r\sqrt{p_{R\vert B}(r\vert b)}\vert u^{r\vert b}\rangle_A\vert v^{r\vert b}\rangle_R.
\end{align}
Correspondingly, we can decompose $\sigma_{AB}$ as
\begin{align}
  \sigma_{AB} &= \Delta_B(\rho_{AB}) = \sum_{b,a} p_{BR}(b,a)\proj{b}_B\ox\proj{u^{a\vert b}}_A,
\end{align}

Consider the Nussbaum and Szko\l{}a's distributions for the operators $\sigma_{BR}$ and $\1_B \ox \sigma_R$:
\begin{subequations}\label{eq: connection reduction proof tmp2}
\begin{align}
    P_{\sigma_{BR},\1_B \ox \sigma_R}((b',r'),(b,r))
&:= P_{BR}(b,r') |\langle v_{r'}^{b'}|v_r\rangle |^2 \delta_{bb'}
  = P_{BR}(b,r') |\langle v_{r'}^{b}|v_r\rangle |^2, \\
    Q_{\sigma_{BR},\1_B \ox \sigma_R}((b',r'),(b,r))
&:= q_r |\langle v_{r'}^{b'}|v_r\rangle |^2 \delta_{bb'}
  = q_r |\langle v_{r'}^{b}|v_r\rangle |^2,
\end{align}
\end{subequations}
and the Nussbaum and Szko\l{}a's distributions for the operators $\rho_{AB}$ and $\sigma_{AB}$:
\begin{subequations}\label{eq: connection reduction proof tmp3}
\begin{align}
      P_{\rho_{AB},\sigma_{AB}}(r,b,a) &= q_r |\langle u_r|b, \psi_a^b\rangle |^2, \\
    Q_{\rho_{AB},\sigma_{AB}}(r,b,a) &= P_{B,R}(b,a) |\langle u_r|b, \psi_a^b\rangle |^2.
\end{align}
\end{subequations}
Since $\ket{\Psi}_{ABR} = \sum_r \sqrt{q_r}\ket{u^r_{AB}}\ket{v^r_R}$,
it holds $\langle v^r|\Psi\rangle \langle \Psi|v^r\rangle= q_r\proj{u^r}$.
On the other hand, Eqs.~\eqref{eq:purification} and~\eqref{eq:Schmidt-decomposition} together imply that
\begin{align}
  \ket{\psi}_{ABR}
= \sum_{b,a}\sqrt{p_{BR}(b,a)}\vert u^{a\vert b}\rangle_A\ox\ket{b}_B\ox\vert v^{a\vert b}\rangle_R.
\end{align}
This leads to
\begin{align}
  P_{\rho_{AB},\sigma_{AB}}(r,b,a)
&= q_r |\langle u_r|b, \psi_a^b\rangle |^2  \\
&= q_r \langle b, \psi_a^b|u_r\rangle \langle u_r|b, \psi_a^b\rangle  \\
&= \langle b, \psi_a^b,v_r|\psi\rangle \langle \psi|b, \psi_a^b,v_r\rangle  \\
&= \tr[ (|b, \psi_a^b \rangle \langle b, \psi_a^b|\otimes \1_R) |\psi\rangle \langle \psi|]
     |\langle v_r| v_a^b\rangle |^2 \\
&= P_{B,R}(b,a) |\langle v_{a}^{b}|v_r\rangle |^2 \\
&= P_{\sigma_{BR},\1_B \ox \sigma_R}(b,a,r).\label{eq: connection reduction proof tmp1}
\end{align}
Hence we can check that
\begin{align}
  &\; D_s^\varepsilon (P_{\sigma_{BR},\1_B \ox \sigma_R}
  \|Q_{\sigma_{BR},\1_B \ox \sigma_R}) \notag\\
  \stackrel{(a)}{=}&\; \sup\left\{x\sbar P_{\sigma_{BR},\1_B \ox \sigma_R} \{(b,a,r)
  |\log P_{\sigma_{BR},\1_B \ox \sigma_R}(b,a,r)
  - \log Q_{\sigma_{BR},\1_B \ox \sigma_R}(a,r) \leq x\} \leq \varepsilon \right\}\\
  \stackrel{(b)}{=}&\; \sup \left\{x\sbar P_{\sigma_{BR},\1_B \ox \sigma_R} \{(b,a,r)|
  \log P_{B,R}(b,a) - \log q_r \leq x\} \leq \varepsilon \right\}\\
  \stackrel{(c)}{=}&\; \sup \left\{x\sbar P_{\rho_{AB},\Delta_B(\rho_{AB})}
  \{(b,a,r)|\log P_{B,R}(b,a) - \log q_r \leq x\} \leq \varepsilon \right\}\\
  \stackrel{(d)}{=}&\; \sup \left\{x\sbar P_{\rho_{AB},\Delta_B(\rho_{AB})}\{(b,a,r)|
  \log Q_{\rho,\Delta(\rho)}(b,a,r)
  - \log P_{\rho,\Delta(\rho)}(b,a,r) \leq x\} \leq \varepsilon \right\}\\
  \stackrel{(e)}{=}&\; - D_s^{1-\varepsilon }(P_{\rho_{AB},\Delta_B(\rho_{AB})}\|Q_{\rho_{AB},\Delta_B(\rho_{AB})}),
\end{align}
where $(a)$ and $(e)$ follow by definition, $(b)$ follows by~\eqref{eq: connection reduction proof tmp2}, $(c)$ follows
by~\eqref{eq: connection reduction proof tmp1}, and $(d)$ follows by~\eqref{eq: connection reduction proof tmp3}. This
concludes~\eqref{eq: reduction connection 1}.
\hfill $\blacksquare$

\noindent\textbf{[Proof of Eqs.~\eqref{eq: reduction connection 2} and~\eqref{eq: reduction connection 3}]}
We first show~\eqref{eq: reduction connection 2}. Consider the following chain of equalities:
\begin{align}
    D(\sigma_{BR}\|\1_B \ox \sigma_R)
    &\stackrel{(a)}{=} D(P_{\sigma_{BR},\1_B \ox \sigma_R}\|Q_{\sigma_{BR},\1_B \ox \sigma_R})\\
    &\stackrel{(b)}{=} \sum_{b,r} P_{\sigma_{BR},\1_B \ox \sigma_R}(b,r) [\log P_{\sigma_{BR},\1_B \ox \sigma_R}(b,r)
    - \log Q_{\sigma_{BR},\1_B \ox \sigma_R}(b,r)]\\
    &\stackrel{(c)}{=} \sum_{b,r} P_{\rho_{AB},\sigma_{AB}}(r,a) [\log p_a - \log q_r]\\
    &\stackrel{(d)}{=} \sum_{b,r} P_{\rho_{AB},\sigma_{AB}}(r,a) [\log Q_{\rho_{AB},\sigma_{AB}}(r,a)
    - \log P_{\rho_{AB},\sigma_{AB}}(r,a)]\\
    &\stackrel{(e)}{=} - D(P_{\rho_{AB},\sigma_{AB}}\|Q_{\rho_{AB},\sigma_{AB}})\\
    &\stackrel{(f)}{=} - D(\rho_{AB}\|\sigma_{AB}),
\end{align}
where $(a)$ and $(f)$ follow from the property of Nussbaum and Szko\l{}a's
distributions~\eqref{eq: PQ first and second order},
 $(b)$ and $(e)$ follow by definition,
  $(c)$ and $(d)$ follow by~\eqref{eq: connection reduction proof tmp2},~
 \eqref{eq: connection reduction proof tmp3},~\eqref{eq: connection reduction proof tmp1}.
This completes the proof of~\eqref{eq: reduction connection 2}.
We can then prove~\eqref{eq: reduction connection 3} in a similar way.
\hfill $\blacksquare$

\noindent\textbf{[Proof of Eq.~\eqref{eq: Hmin DH}]}
First notice that $\sigma_R$ and $\rho_{AB}$ are two marginal of the pure tripartite state $\psi_{ABR}$, thus they
have the same eigenvalues. This gives $\theta(\sigma_R) = \theta(\rho_{AB})$.
Consider the following chain of inequalities:
\begin{align}
    H_{\min}^{{\varepsilon }}(B|R)_{\sigma}
:=&\; - \inf_{\tau_R}D_{\max}^{{\varepsilon }}(\sigma_{BR}\|\1_B \ox \tau_R) \\
\geq&\;  - D_{\max}^{{\varepsilon }}(\sigma_{BR}\|\1_B \ox \sigma_R)\\
\stackrel{(a)}{\geq}&\;  - D_s^{1-\varepsilon ^2+\delta}(P_{\sigma_{BR},\1_B \ox \sigma_{R}}
\|Q_{\sigma_{BR},\1_B \ox \sigma_{R}}) - \log\theta(\rho_{AB}) + \log(\delta\varepsilon^2)\\
\stackrel{(b)}{=}&\; D_s^{\varepsilon ^2-\delta}(P_{\rho_{AB},\Delta(\rho_{AB})}
\|Q_{\rho_{AB},\Delta(\rho_{AB})}) - \log\theta(\rho_{AB}) + \log(\delta\varepsilon^2)\\
\stackrel{(c)}{\geq}&\; D_H^{\varepsilon ^2-2\delta}(\rho_{AB}\|\Delta_B(\rho_{AB})) - c(\rho_{AB},\varepsilon,\delta),
\end{align}
where $(a)$ follows from~\cite[Eq. (29)]{tomamichel2013hierarchy} and the fact that $\theta (\1_A
\ox \sigma_R) = \theta(\sigma_R)= \theta(\rho_{AB})$, $(b)$ follows from~\eqref{eq: reduction
connection 1}, and $(c)$ follows from~\cite[Eq. (27)]{tomamichel2013hierarchy} with the correction term
$c(\rho_{AB},\varepsilon,\delta)$ given in~\eqref{eq:correction term 2}. Note that the constraints
$\delta<1-\varepsilon ^2$ and $\delta<\varepsilon^2/3$ are imposed in $(a)$ and $(c)$, respectively. This completes the proof.

\hfill $\blacksquare$

\section{Conclusions}

Our work initiated the \emph{first} systematic second order analysis on coherence distillation with and without assistance, filling an important gap in the literature. In the unassisted setting,
we introduced a variant of randomness extraction framework in the context of quantum coherence theory, establishing an exact relation between this cryptographic task and the operational task of quantum coherence distillation. Based on this relation, we gave a finite block length analysis on these tasks, providing in particular explicit second order expansions of distillable coherence and extractable randomness under a diverse range of free operations. We then lifted the obtained results to the assisted setting in which Alice served as an assistant to help Bob do the manipulations. A crucial step for our second order expansions was the hypothesis testing characterizations of the one-shot rate with {almost tight} error dependence. These one-shot characterizations could be suitable for analysis even beyond the i.i.d. assumption of the source state if combining with a more refined result of hypothesis testing relative entropy (e.g.~\cite{datta2016second} and the references therein).

Many interesting problems remain open. First, in the unassisted setting
the coincidence of the second order expansions of $\ell_{\cO}^\ve$ and $\ell_{\id}^\ve$ indicates that optimizing the free incoherent operations before the incoherent measurement can improve the extractable randomness by the order $O(\log n)$ at most.
One may explore whether there is any advantage of performing incoherent operations in the third or higher order terms.
Second, as a reverse problem of coherence distillation, the coherence cost considers the minimum number of
coherent bits required to prepare a quantum state. It is known that the first order asymptotics of
coherence cost under $\IO$ operations is given by the coherence information~\cite{winter_2016}.
But what is the second order expansion?
Recall the important role of randomness extraction framework in our second order analysis.
For coherence cost, we may consider a randomness extraction scenario with Eve having limited power.
Such a scenario has been studied in~\cite{Yuan} and \cite[Section VI]{hayashi2018secure} and
the corresponding randomness extraction rate happens to coincide with the coherence information.
Finally, regarding the assisted scenario  it would be interesting to extend the exact one-shot relation between $C_{d,\QIP}^\varepsilon$ and $\ell_{\QIP}^\varepsilon$ to other free operation classes.
It is also appealing to further explore the alternative formulation in Appendix~\ref{sec:alternative formulation} and identify
achievable rate in terms of hypothesis testing relative entropy with $\varepsilon$-square error dependence.


\paragraph{Acknowledgements.} 
We thank Anurag Anshu and Xin Wang for discussions about the second order expansion of distillable coherence. KF thanks the Center for Quantum Computing at Peng Cheng Laboratory for their hospitality while part of this work was done during his visit. MH was supported in part by a JSPS Grant-in-Aid for Scientific Research (A) No.17H01280, (B) No. 16KT0017, the Okawa Research Grant and Kayamori Foundation of Informational Science Advancement.

\bibliographystyle{alpha_abbrv}
{\small \bibliography{bib}}

\newcommand{\etalchar}[1]{$^{#1}$}
\begin{thebibliography}{KMWY16b}

\bibitem[Abe06]{aberg2006quantifying}
J.~Aberg.
\newblock Quantifying superposition.
\newblock {\em arXiv preprint quant-ph/0612146}, 2006.

\bibitem[AJS18]{anshu2018quantum}
A.~Anshu, R.~Jain, and A.~Streltsov.
\newblock Quantum state redistribution with local coherence.
\newblock {\em arXiv:1804.04915}, 2018.

\bibitem[BCD05]{buscemi2005inverting}
F.~Buscemi, G.~Chiribella, and G.~M. D’Ariano.
\newblock Inverting quantum decoherence by classical feedback from the
  environment.
\newblock {\em Physical Review Letters}, 95(9):090501, 2005.

\bibitem[BCP14]{baumgratz2014quantifying}
T.~Baumgratz, M.~Cramer, and M.~B. Plenio.
\newblock Quantifying coherence.
\newblock {\em Physical Review Letters}, 113(14):140401, 2014.

\bibitem[BD13]{buscemi2013general}
F.~Buscemi and N.~Datta.
\newblock General theory of environment-assisted entanglement distillation.
\newblock {\em IEEE Transactions on Information Theory}, 59(3):1940--1954,
  2013.

\bibitem[Bus07]{buscemi2007channel}
F.~Buscemi.
\newblock Channel correction via quantum erasure.
\newblock {\em Physical Review Letters}, 99(18):180501, 2007.

\bibitem[CG16]{chitambar_2016}
E.~Chitambar and G.~Gour.
\newblock Critical examination of incoherent operations and a physically
  consistent resource theory of quantum coherence.
\newblock {\em Physical Review Letters}, 117:030401, 2016.

\bibitem[CG19]{chitambar2018quantum}
E.~Chitambar and G.~Gour.
\newblock Quantum resource theories.
\newblock {\em Reviews of Modern Physics}, 91(2):025001, 2019.

\bibitem[CH16]{chitambar_2016-2}
E.~Chitambar and M.-H. Hsieh.
\newblock Relating the {{Resource Theories}} of {{Entanglement}} and {{Quantum
  Coherence}}.
\newblock {\em Physical Review Letters}, 117:020402, 2016.

\bibitem[CML16]{coles2016numerical}
P.~J. Coles, E.~M. Metodiev, and N.~L{\"u}tkenhaus.
\newblock Numerical approach for unstructured quantum key distribution.
\newblock {\em Nature Communications}, 7:11712, 2016.

\bibitem[CSR{\etalchar{+}}16a]{chitambar_2016-3}
E.~Chitambar, A.~Streltsov, S.~Rana, M.~Bera, G.~Adesso, and M.~Lewenstein.
\newblock Assisted {{Distillation}} of {{Quantum Coherence}}.
\newblock {\em Physical Review Letters}, 116:070402, 2016.

\bibitem[CSR{\etalchar{+}}16b]{chitambar2016assisted}
E.~Chitambar, A.~Streltsov, S.~Rana, M.~Bera, G.~Adesso, and M.~Lewenstein.
\newblock Assisted distillation of quantum coherence.
\newblock {\em Physical Review Letters}, 116(7):070402, 2016.

\bibitem[Dat09]{datta2009min}
N.~Datta.
\newblock Min-and max-relative entropies and a new entanglement monotone.
\newblock {\em IEEE Transactions on Information Theory}, 55(6):2816--2826,
  2009.

\bibitem[DFM{\etalchar{+}}98]{divincenzo1998entanglement}
D.~P. DiVincenzo, C.~A. Fuchs, H.~Mabuchi, J.~A. Smolin, A.~Thapliyal, and
  A.~Uhlmann.
\newblock Entanglement of assistance.
\newblock In {\em NASA International Conference on Quantum Computing and
  Quantum Communications}, pages 247--257. Springer, 1998.

\bibitem[DFW{\etalchar{+}}18]{diaz2018}
M.~G. D{\'{i}}az, K.~Fang, X.~Wang, M.~Rosati, M.~Skotiniotis, J.~Calsamiglia,
  and A.~Winter.
\newblock Using and reusing coherence to realize quantum processes.
\newblock {\em {Quantum}}, 2:100, October 2018.

\bibitem[DH10]{dutil2010assisted}
N.~Dutil and P.~Hayden.
\newblock Assisted entanglement distillation.
\newblock {\em arXiv preprint arXiv:1011.1972}, 2010.

\bibitem[DL14]{datta2014second}
N.~Datta and F.~Leditzky.
\newblock Second-order asymptotics for source coding, dense coding, and
  pure-state entanglement conversions.
\newblock {\em IEEE Transactions on Information Theory}, 61(1):582--608, 2014.

\bibitem[DPR16]{datta2016second}
N.~Datta, Y.~Pautrat, and C.~Rouz{\'e}.
\newblock Second-order asymptotics for quantum hypothesis testing in settings
  beyond iid—quantum lattice systems and more.
\newblock {\em Journal of Mathematical Physics}, 57(6):062207, 2016.

\bibitem[DVS16]{vicente_2017}
J.~I. De~Vicente and A.~Streltsov.
\newblock Genuine quantum coherence.
\newblock {\em Journal of Physics A: Mathematical and Theoretical},
  50(4):045301, 2016.

\bibitem[FD11]{Frowis2011}
F.~Fr{\"{o}}wis and W.~D{\"{u}}r.
\newblock {Stable Macroscopic Quantum Superpositions}.
\newblock {\em Physical Review Letters}, 106(11):110402, mar 2011.

\bibitem[FWL{\etalchar{+}}18]{fang2018probabilistic}
K.~Fang, X.~Wang, L.~Lami, B.~Regula, and G.~Adesso.
\newblock Probabilistic distillation of quantum coherence.
\newblock {\em Physical Review Letters}, 121(7):070404, 2018.

\bibitem[FWTD19]{fang2019non}
K.~Fang, X.~Wang, M.~Tomamichel, and R.~Duan.
\newblock Non-asymptotic entanglement distillation.
\newblock {\em IEEE Transactions on Information Theory}, 65(10):6454--6465,
  2019.

\bibitem[GS08]{gour_2008}
G.~Gour and R.~W. Spekkens.
\newblock The resource theory of quantum reference frames: manipulations and
  monotones.
\newblock {\em New Journal of Physics}, 10(3):033023, 2008.

\bibitem[GW03]{gregoratti2003quantum}
M.~Gregoratti and R.~F. Werner.
\newblock Quantum lost and found.
\newblock {\em Journal of Modern Optics}, 50(6-7):915--933, 2003.

\bibitem[Hay06]{hayashi2006practical}
M.~Hayashi.
\newblock Practical evaluation of security for quantum key distribution.
\newblock {\em Physical Review A}, 74(2):022307, 2006.

\bibitem[Hay08]{hayashi2008second}
M.~Hayashi.
\newblock Second-order asymptotics in fixed-length source coding and intrinsic
  randomness.
\newblock {\em IEEE Transactions on Information Theory}, 54(10):4619--4637,
  2008.

\bibitem[HHHH09]{horodecki2009quantum}
R.~Horodecki, P.~Horodecki, M.~Horodecki, and K.~Horodecki.
\newblock Quantum entanglement.
\newblock {\em Reviews of Modern Physics}, 81(2):865, 2009.

\bibitem[Hil16]{hillery2016coherence}
M.~Hillery.
\newblock Coherence as a resource in decision problems: The {Deutsch-Jozsa}
  algorithm and a variation.
\newblock {\em Physical Review A}, 93(1):012111, 2016.

\bibitem[HK04]{hayden2004correcting}
P.~Hayden and C.~King.
\newblock Correcting quantum channels by measuring the environment.
\newblock {\em arXiv preprint quant-ph/0409026}, 2004.

\bibitem[HZ18]{hayashi2018secure}
M.~Hayashi and H.~Zhu.
\newblock Secure uniform random-number extraction via incoherent strategies.
\newblock {\em Physical Review A}, 97(1):012302, 2018.

\bibitem[KMWY16a]{karumanchi2016classical}
S.~Karumanchi, S.~Mancini, A.~Winter, and D.~Yang.
\newblock Classical capacities of quantum channels with environment assistance.
\newblock {\em Problems of Information Transmission}, 52(3):214--238, 2016.

\bibitem[KMWY16b]{karumanchi2016quantum}
S.~Karumanchi, S.~Mancini, A.~Winter, and D.~Yang.
\newblock Quantum channel capacities with passive environment assistance.
\newblock {\em IEEE Transactions on Information Theory}, 62(4):1733--1747,
  2016.

\bibitem[Lam19]{lami2019completing}
L.~Lami.
\newblock Completing the grand tour of asymptotic quantum coherence
  manipulation.
\newblock {\em IEEE Transactions on Information Theory}, to appear, 2019.

\bibitem[Li14]{li2014second}
K.~Li.
\newblock Second-order asymptotics for quantum hypothesis testing.
\newblock {\em The Annals of Statistics}, 42(1):171--189, 2014.

\bibitem[LJR15]{lostaglio2015description}
M.~Lostaglio, D.~Jennings, and T.~Rudolph.
\newblock Description of quantum coherence in thermodynamic processes requires
  constraints beyond free energy.
\newblock {\em Nature Communications}, 6:6383, 2015.

\bibitem[LM14]{levi_2014}
F.~Levi and F.~Mintert.
\newblock A quantitative theory of coherent delocalization.
\newblock {\em New Journal of Physics}, 16(3):033007, 2014.

\bibitem[LRA19]{lami2019generic}
L.~Lami, B.~Regula, and G.~Adesso.
\newblock Generic bound coherence under strictly incoherent operations.
\newblock {\em Physical Review Letters}, 122(15):150402, 2019.

\bibitem[LTA20]{lami2020assisted}
L.~Lami, R.~Takagi, and G.~Adesso.
\newblock Assisted concentration of gaussian resources.
\newblock {\em Physical Review A}, 101(5):052305, 2020.

\bibitem[LWZG09]{li2009private}
K.~Li, A.~Winter, X.~Zou, and G.~Guo.
\newblock Private capacity of quantum channels is not additive.
\newblock {\em Physical Review Letters}, 103(12):120501, 2009.

\bibitem[MS16]{marvian_2016}
I.~Marvian and R.~W. Spekkens.
\newblock How to quantify coherence: {{Distinguishing}} speakable and
  unspeakable notions.
\newblock {\em Physical Review A}, 94:052324, 2016.

\bibitem[NS09]{nussbaum2009}
M.~Nussbaum and A.~Szko{\l}a.
\newblock The chernoff lower bound for symmetric quantum hypothesis testing.
\newblock {\em The Annals of Statistics}, 37(2):1040--1057, 2009.

\bibitem[Ren05]{Renner2005}
R.~Renner.
\newblock {Security of Quantum Key Distribution}.
\newblock PhD thesis, dec 2005.

\bibitem[RFWA18]{regula2018one}
B.~Regula, K.~Fang, X.~Wang, and G.~Adesso.
\newblock One-shot coherence distillation.
\newblock {\em Physical Review Letters}, 121(1):010401, 2018.

\bibitem[RLS18]{regula2018nonasymptotic}
B.~Regula, L.~Lami, and A.~Streltsov.
\newblock Nonasymptotic assisted distillation of quantum coherence.
\newblock {\em Physical Review A}, 98(5):052329, 2018.

\bibitem[SAP17]{streltsov_2017}
A.~Streltsov, G.~Adesso, and M.~B. Plenio.
\newblock Quantum coherence as a resource.
\newblock {\em Reviews of Modern Physics}, 89:041003, 2017.

\bibitem[SCR{\etalchar{+}}16]{streltsov_2016}
A.~Streltsov, E.~Chitambar, S.~Rana, M.~Bera, A.~Winter, and M.~Lewenstein.
\newblock Entanglement and {{Coherence}} in {{Quantum State Merging}}.
\newblock {\em Physical Review Letters}, 116:240405, 2016.

\bibitem[SP00]{shor2000simple}
P.~W. Shor and J.~Preskill.
\newblock Simple proof of security of the bb84 quantum key distribution
  protocol.
\newblock {\em Physical Review Letters}, 85(2):441, 2000.

\bibitem[SRBL17]{streltsov2017towards}
A.~Streltsov, S.~Rana, M.~N. Bera, and M.~Lewenstein.
\newblock Towards resource theory of coherence in distributed scenarios.
\newblock {\em Physical Review X}, 7(1):011024, 2017.

\bibitem[SVW05]{smolin2005entanglement}
J.~A. Smolin, F.~Verstraete, and A.~Winter.
\newblock Entanglement of assistance and multipartite state distillation.
\newblock {\em Physical Review A}, 72(5):052317, 2005.

\bibitem[TBR16]{tomamichel2016quantum}
M.~Tomamichel, M.~Berta, and J.~M. Renes.
\newblock Quantum coding with finite resources.
\newblock {\em Nature communications}, 7(1):1--8, 2016.

\bibitem[TH13]{tomamichel2013hierarchy}
M.~Tomamichel and M.~Hayashi.
\newblock A hierarchy of information quantities for finite block length
  analysis of quantum tasks.
\newblock {\em IEEE Transactions on Information Theory}, 59(11):7693--7710,
  2013.

\bibitem[TLGR12]{tomamichel2012tight}
M.~Tomamichel, C.~C.~W. Lim, N.~Gisin, and R.~Renner.
\newblock Tight finite-key analysis for quantum cryptography.
\newblock {\em Nature communications}, 3(1):1--6, 2012.

\bibitem[Tom15]{tomamichel2015quantum}
M.~Tomamichel.
\newblock {\em Quantum Information Processing with Finite Resources:
  Mathematical Foundations}, volume~5.
\newblock Springer, 2015.

\bibitem[Uhl76]{uhlmann1976transition}
A.~Uhlmann.
\newblock The “transition probability” in the state space of *-algebra.
\newblock {\em Reports on Mathematical Physics}, 9(2):273--279, 1976.

\bibitem[VCH18]{vijayan2018one}
M.~K. Vijayan, E.~Chitambar, and M.-H. Hsieh.
\newblock One-shot assisted concentration of coherence.
\newblock {\em Journal of Physics A: Mathematical and Theoretical},
  51(41):414001, 2018.

\bibitem[WFT19]{wang2019converse}
X.~Wang, K.~Fang, and M.~Tomamichel.
\newblock On converse bounds for classical communication over quantum channels.
\newblock {\em IEEE Transactions on Information Theory}, 65(7):4609--4619,
  2019.

\bibitem[WHZ{\etalchar{+}}17]{wu2017experimentally}
K.-D. Wu, Z.~Hou, H.-S. Zhong, Y.~Yuan, G.-Y. Xiang, C.-F. Li, and G.-C. Guo.
\newblock Experimentally obtaining maximal coherence via assisted distillation
  process.
\newblock {\em Optica}, 4(4):454--459, 2017.

\bibitem[WHZ{\etalchar{+}}18]{wu2018experimental}
K.-D. Wu, Z.~Hou, Y.-Y. Zhao, G.-Y. Xiang, C.-F. Li, G.-C. Guo, J.~Ma, Q.-Y.
  He, J.~Thompson, and M.~Gu.
\newblock Experimental cyclic interconversion between coherence and quantum
  correlations.
\newblock {\em Physical Review Letters}, 121(5):050401, 2018.

\bibitem[Win05]{winter2005environment}
A.~Winter.
\newblock On environment-assisted capacities of quantum channels.
\newblock {\em arXiv preprint quant-ph/0507045}, 2005.

\bibitem[WR12]{wang2012one}
L.~Wang and R.~Renner.
\newblock One-shot classical-quantum capacity and hypothesis testing.
\newblock {\em Physical Review Letters}, 108(20):200501, 2012.

\bibitem[WTX{\etalchar{+}}20]{wu2020quantum}
K.-D. Wu, T.~Theurer, G.-Y. Xiang, C.-F. Li, G.-C. Guo, M.~B. Plenio, and
  A.~Streltsov.
\newblock Quantum coherence and state conversion: theory and experiment.
\newblock {\em npj Quantum Information}, 6(1):1--9, 2020.

\bibitem[WY16]{winter_2016}
A.~Winter and D.~Yang.
\newblock Operational resource theory of coherence.
\newblock {\em Physical Review Letters}, 116:120404, 2016.

\bibitem[YHW19]{yang2019distributed}
D.~Yang, K.~Horodecki, and A.~Winter.
\newblock Distributed private randomness distillation.
\newblock {\em Physical Review Letters}, 123(17):170501, 2019.

\bibitem[YVH19]{yamasaki2019hierarchy}
H.~Yamasaki, M.~K. Vijayan, and M.-H. Hsieh.
\newblock Hierarchy of quantum operations in manipulating coherence and
  entanglement.
\newblock {\em arXiv preprint arXiv:1912.11049}, 2019.

\bibitem[YZCM15]{Yuan}
X.~Yuan, H.~Zhou, Z.~Cao, and X.~Ma.
\newblock Intrinsic randomness as a measure of quantum coherence.
\newblock {\em Physical Review A}, 92(2):022124, 2015.

\bibitem[ZLY{\etalchar{+}}18]{zhao_2018}
Q.~Zhao, Y.~Liu, X.~Yuan, E.~Chitambar, and X.~Ma.
\newblock One-{{Shot Coherence Dilution}}.
\newblock {\em Physical Review Letters}, 120:070403, 2018.

\bibitem[ZLY{\etalchar{+}}19]{zhao2019one}
Q.~Zhao, Y.~Liu, X.~Yuan, E.~Chitambar, and A.~Winter.
\newblock One-shot coherence distillation: Towards completing the picture.
\newblock {\em IEEE Transactions on Information Theory}, 65(10):6441--6453,
  2019.

\bibitem[ZMF17]{zhao2017coherence}
M.-J. Zhao, T.~Ma, and S.-M. Fei.
\newblock Coherence of assistance and regularized coherence of assistance.
\newblock {\em Physical Review A}, 96(6):062332, 2017.

\bibitem[ZMQ{\etalchar{+}}19]{zhao20191}
M.-J. Zhao, T.~Ma, Q.~Quan, H.~Fan, and R.~Pereira.
\newblock {$l_1$}-norm coherence of assistance.
\newblock {\em Physical Review A}, 100(1):012315, 2019.

\end{thebibliography}

\newpage

\begin{appendices}

\section{Proof of Proposition~\ref{prop: IO protocol}}
\label{sec:proposition-one-proof}

\noindent \textbf{[Proof]}
Let $\sigma^*_{R}$ be a quantum state that attains the minimum in
\begin{align}
    d_{sec}(\rho[\id,\Delta,f]_{LR}|R) = \min_{\sigma_R \in \cS(R)} P(\rho[\id,\Delta,f]_{LR},\pi_L\ox \sigma_R).
\end{align}
Let $\ket{\phi^*}_{BR}$ on $\cH_B \ox \cH_R$ be a purification of $\sigma^*_{R}$. Thus we have
\begin{align}\label{eq:direct-tmp1}
    F(\rho[\id,\Delta,f]_{LR},\pi_L\ox \sigma^*_{R}) \geq \sqrt{1-\ve^2}.
\end{align}
Define the incoherent isometry $U_f$ from $\cH_B$ to $\cH_L \ox \cH_B$ as
\begin{align}
U_f \ket{b}_B:= \ket{f(b)}_L \ox \ket{b}_B.
\end{align}
We choose normalized vectors $\ket{\phi_\ell}_{BR}$ and normalization factors $r_\ell$ such that
\begin{align}
    U_f \ket{\psi}_{BR} = \sum_{\ell\in \cL} \sqrt{r_\ell} \ket{\ell}_L \ox \ket{\phi_\ell}_{BR}
    \quad \text{and}\quad
    \sqrt{r_\ell} \ket{\phi_\ell}_{BR}:= \sum_{b \in \cB: f(b) = \ell} \sqrt{p_b}\ket{b}_B\ox \ket{\psi_b}_R.
\end{align}
By Uhlmann's theorem~\cite{uhlmann1976transition} there exists a unitary $U_\ell$ on $\cH_B$ such that
\begin{align}\label{eq:direct-tmp2}
  F(\tr_B \ket{\phi_\ell}\bra{\phi_\ell}_{BR},\sigma^*_{R})
= F(U_\ell \ket{\phi_\ell}_{BR},\ket{\phi^*}_{BR}) = \<\phi^*|U_\ell|\phi_\ell\> .
\end{align}
Take $U := \sum_{\ell \in \cL} \ket{\ell}\bra{\ell}_L \ox U_\ell$.
We have
\begin{align}
F\left(UU_f \ket{\psi}_{BR}, \ket{\Psi_L} \ox \ket{\phi^*}_{BR}\right)
& \stackrel{(a)}{=} F\Big(\sum_{\ell\in \cL} \sqrt{r_\ell} \ket{\ell} \ox U_\ell \ket{\phi_\ell}, \sum_{\ell\in\cL} \frac{1}{\sqrt{|L|}} \ket{\ell} \ox \ket{\phi^*}\Big)\\
& \stackrel{(b)}{=} \sum_{\ell\in \cL} \sqrt{r_\ell}\frac{1}{\sqrt{|L|}} F(U_\ell \ket{\phi_\ell},\ket{\phi^*})\\
& \stackrel{(c)}{=} \sum_{\ell\in\cL} \sqrt{r_\ell}\frac{1}{\sqrt{|L|}} F(\tr_B \ket{\phi_\ell}\bra{\phi_\ell},\sigma_R^*)\\
& \stackrel{(d)}{=} F\Big(\sum_{\ell\in \cL} r_\ell \ket{\ell}\bra{\ell}\ox \tr_B \ket{\phi_\ell}\bra{\phi_\ell}, \sum_{\ell\in\cL} \frac{1}{|L|} \ket{\ell}\bra{\ell}\ox \sigma_R^*\Big)\\
& \stackrel{(e)}{=} F(\rho[\id,\Delta,f]_{LR}, \pi_L\ox \sigma^*_{R}), \label{eq:direct-tmp5}
\end{align}
where $(a)$ follows by definition, $(b)$ and $(d)$ follow from Lemma~\ref{lem: fidelity}, $(c)$ follows from Eq.~\eqref{eq:direct-tmp2} and $(e)$ follows from the fact that $\rho[\id,\Delta,f]_{LR} = \sum_{\ell\in \cL} r_\ell \ket{\ell}\bra{\ell}\ox \tr_B \ket{\phi_\ell}\bra{\phi_\ell}$.
We construct a quantum operation $\Gamma_{B\to L}(\cdot) := \tr_B [UU_f (\cdot) U_f^\dagger U^\dagger]$, whose
schematic diagram is given in Figure~\ref{fig: protocol of coherence distillation}.
Then $\Gamma_{B\to L}(\rho_B) = \tr_{BR}[UU_f \ket{\psi}\bra{\psi}_{BR} U_f^\dagger U^\dagger]$ and we can check that
\begin{align}
F \left(\Gamma_{B\to L}(\rho_B), \Psi_L\right) \geq F\big(UU_f \ket{\psi}_{BR}, \ket{\Psi_L} \ox \ket{\phi^*}_{BR}\big) = F\big(\rho[\id,\Delta,f]_{LR},\pi_L\ox \sigma^*_{R}\big) \geq \sqrt{1-\ve^2},
\end{align}
where the first inequality follows by the data-processing inequality of quantum fidelity under $\tr_{BR}$, the equality follows by~\eqref{eq:direct-tmp5} and the second inequality follows from~\eqref{eq:direct-tmp1}. This implies that $P(\Gamma_{B\to L}(\rho_B), \Psi_L ) \leq \ve$.

It remains to check $\Gamma \in \DIIO$. Note that $\Gamma$ admits a Kraus decomposition $\Gamma(\cdot) = \sum_{b\in
\cB} K_b U_f (\cdot) U_f^\dagger K_b^\dagger$ with operators $K_b = \bra{b}U$. For any computational basis $\ket{x}$
and any $b \in \cB$, we have
\begin{align}
K_b U_f \ket{x}\bra{x} U_f^\dagger K_b^\dagger  = |\<b|U_{f(x)}|x\>|^2\ket{f(x)}\bra{f(x)} \in \cI^{**},
\end{align}
by direct calculation.
Thus $\Gamma \in \IO$. For any computational basis $\ket{x}$ and $\ket{y}$, we can first check that
\begin{align}
\Gamma(\ket{x}\bra{y}) = \<y|U_{f(x)}^\dagger U_{f(y)}|x\>  \ket{f(x)}\bra{f(y)}.
\end{align}
Thus it holds
\begin{align}\label{eq:direct-tmp3}
\Gamma(\Delta(\ket{x}\bra{y})) & = \Gamma(\delta_{x,y} \ket{x}\bra{x}) = \delta_{x,y} \Gamma(\ket{x}\bra{x}) = \delta_{x,y} \ket{f(x)}\bra{f(x)},
\end{align}
and
\begin{align}\label{eq:direct-tmp4}
\Delta(\Gamma(\ket{x}\bra{y})) = \delta_{f(x),f(y)}\<y|U_{f(x)}^\dagger U_{f(y)}|x\> \ket{f(x)}\bra{f(x)} =  \delta_{f(x),f(y)} \delta_{x,y}\ket{f(x)}\bra{f(x)} = \delta_{x,y} \ket{f(x)} \bra{f(x)}.
\end{align}
Combining~\eqref{eq:direct-tmp3} and~\eqref{eq:direct-tmp4}, we have $\Gamma \circ \Delta = \Delta \circ \Gamma$, indicating that $\Gamma \in \DIO$. Finally we have $\Gamma \in \DIIO$.
\hfill $\blacksquare$

\begin{figure}[t]
\centering
\includegraphics[width=0.6\textwidth]{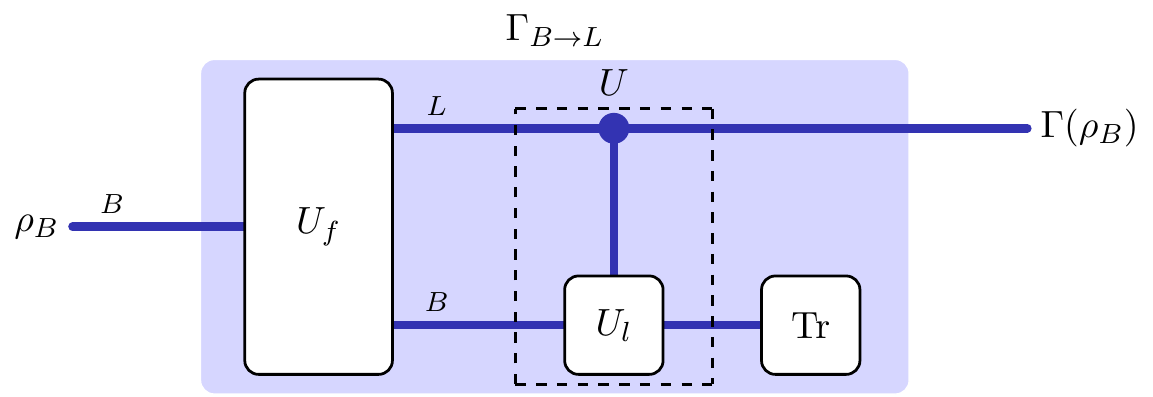}
\caption{Schematic diagram of the $\DIIO$ distillation protocol in 
        Proposition~\ref{prop: IO protocol}.
        The isometry $U_f$ is incoherent.
        The unitary $U$, represented by the inner dashed box, is a controlled unitary but is not
        necessarily incoherent. 
        The overall operation $\Gamma_{B\to L}$, represented by the dashed box, belongs to DIIO.}
\label{fig: protocol of coherence distillation}
\end{figure}

\section{An alternative formulation of assisted incoherent randomness extraction}
\label{sec:alternative formulation}

In this appendix, we consider an alternative formulation of the assisted incoherent randomness extraction framework
that is different from the one discussed in Section~\ref{sec:assisted randomness extraction}. In this formulation, Bob
is more discreet in the sense that he accepts the assistance from Alice but does not allow her to possess any
information about the extracted randomness.

\subsection{Task description}

In the beginning, Alice and Bob preshare a bipartite quantum state $\rho_{AB}$ with purification $\ket{\psi}_{RAB}$
such that the reference system $R$ held by Eve. A general assisted incoherent randomness extraction protocol is
characterized by a triplet $(\Lambda,\Delta,f)$, where $\Lambda\in\cF$ a
free operation and $f$ a hash function. The protocol has three steps:
\begin{enumerate}
\item Alice and Bob first perform a free operation $\Lambda_{AB\to C}\in\cF$ on their joint system, where $C$ is in
  Bob's possession. That is, there is no output system at Alice's hand any more. This can be done by performing a
  partial trace operation. Let $U_{AB\to CE}$ be a Stinespring isometry representation of $\Lambda$. We assume the
  environment system $E$ of $\Lambda$ is also controlled by Eve. After the action of $\Lambda$, the whole system is in
  a pure state
  \begin{align}\label{eq:rho-Lambda-CER}
    \tau[\Lambda]_{CER} := U_{AB\to CE}(\proj{\psi}_{RAB})U_{AB\to CE}^\dagger,
  \end{align}
  where we use $\tau$ instead of $\rho$ in~\eqref{eq:rho-Lambda-CER} to indicate that we consider an alternative
  formulation here.
\item Bob dephases system $C$ into the incoherent basis via $\Delta_C$.
\item Bob performs the hash function $f$ to extract randomness. 
  These two steps lead to the final output state
  \begin{align}
    \tau[\Lambda,\Delta,f]_{LER} := \tau[\Lambda,\Delta]_{f(C)ER}
    = \sum_{b\in\cB}p_b\proj{f(b)}_{L}\ox\sigma_{ER}^b,
  \end{align}
  where $p_b := \tr\bra{b}\tau[\Lambda]_{CER}\ket{b}$ and $\sigma_{ER}^b:=\bra{b}\tau[\Lambda]_{CER}\ket{b}/p_b$.
\end{enumerate} 
A detailed procedure of this alternative assisted randomness extraction via $(\Lambda,\Delta,f)$ is depicted in
Figure~\ref{fig:assisted-extraction-alt}.
Likewise, the \emph{one-shot assisted extractable randomness via $\cF$}
of the quantum state
$\rho_{AB}$ in this alternative formulation is defined as
\begin{align}\label{eq:one-shot-er-alt}
    \widetilde{\ell}_{\cF}^\varepsilon(\rho_{AB}) :=
    \max_{\Lambda_{AB\to C}\in\cF}
    \max_{f}\left\{\log |L| :
    d_{sec}\left(\tau[\Lambda,\Delta,f]_{LER}|ER\right)\leq \varepsilon \right\}.
\end{align}

\begin{remark}\label{remark:alternative}
As one can tell, the essential difference between this formulation and the original definition in
Section~\ref{sec:assisted randomness extraction} lies in what kind of incoherent operation $\Lambda\in\cF$ that Alice
and Bob can use. In the original definition, Alice and Bob can adopt operations of the form $\Lambda_{AB\to A^\prime
B^\prime}$ such that the output $A^\prime$ is at Alice's hand and $B^\prime$ is at Bob's hand. As a result, Alice has
some information on Bob's extracted randomness that is secure from Ever. However, in this alternative formulation, we
rule out this possibility by only allowing operations of the form $\Lambda_{AB\to C}$ such that the output system $C$
is under the full control of Bob. Apparently, Bob in the alternative formulation is more stringent since he solely
makes use of Alice while does not allow her to possess any secrecy. However, this difference does not cause
any trouble in the assisted coherence distillation scenario, since $\Lambda_{AB\to A^\prime B^\prime}$ can be converted
to $\Lambda_{AB\to C}$ by doing an extra partial trace $\tr_{A^\prime}$ without affecting the distillation rate.
\end{remark}

\begin{figure}[!hbtp]
\centering
\includegraphics[width=0.8\textwidth]{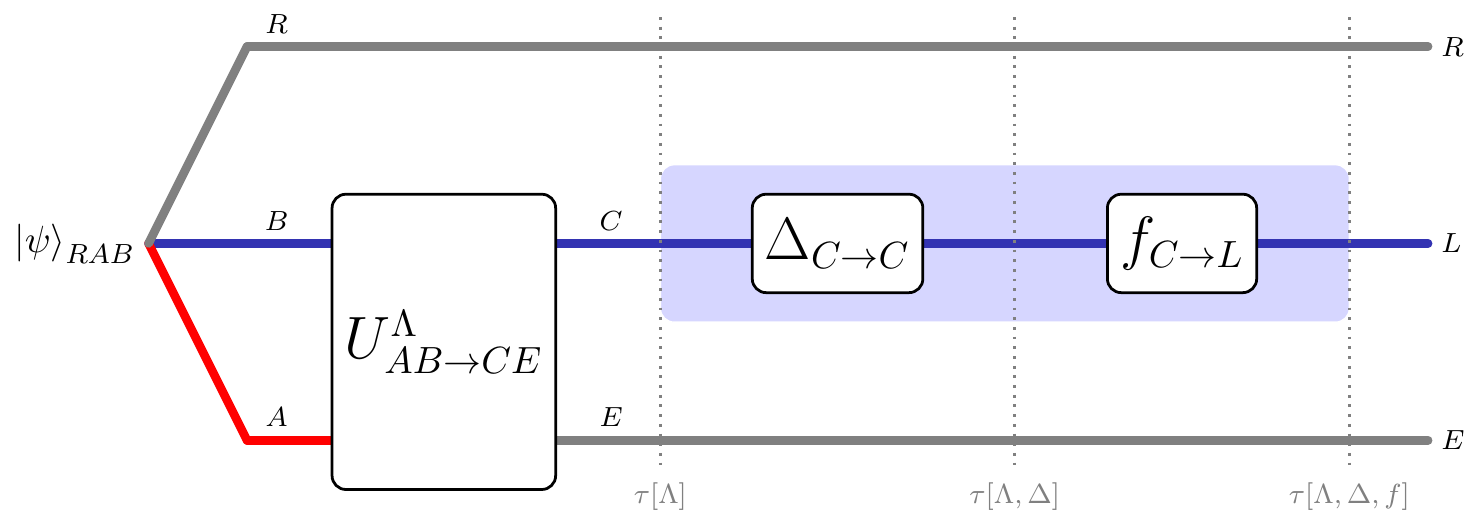}
\caption{Alternative formulation of the assisted randomness extraction framework given by $[\Lambda,\Delta,f]$.
The system in red belongs to Alice, the system in blue belong to Bob, and the two systems in gray belong to Eve. In the
shaded area, we illustrate an incoherent randomness extraction protocol $[\id,\Delta,f]$ (without assistance!) for the
state $\tau_C$ with purification $\tau[\Lambda]_{CER}$.}
\label{fig:assisted-extraction-alt}
\end{figure}

Based on the argument in Remark~\ref{remark:alternative}, we conclude that Bob in the original assisted randomness
extraction framework (cf. Section~\ref{sec:assisted randomness extraction}) can extract more randomness than in this
alternative framework.

\begin{proposition}\label{prop:two-er-relation}
Let $\cF\in\{\LICC,\LQICC,\SI,\SQI,\QIP\}$. Let $\rho_{AB}$ be a bipartite quantum state and $\varepsilon\in[0,1]$. It
holds that
\begin{align}
  \widetilde{\ell}_{\cF}^\varepsilon(\rho_{AB}) \leq \ell_{\cF}^\varepsilon(\rho_{AB}).
\end{align}
\end{proposition}

\subsection{New equivalence relation}

The main advantage of this alternative assisted incoherent randomness extraction framework is that we are able to
establish an equivalence relation between the assisted coherence distillation described in Section~\ref{sec:assisted
coherence distillation} and the alternative assisted incoherent randomness extraction described here, for all free
operation classes under consideration. This equivalence hence is much stronger than that was concluded in
Theorem~\ref{thm:equivalence}, holding only for the $\QIP$ class.

\begin{theorem}[Equivalence relation II]\label{thm:equivalence-alternative}
Let $\cF\in\{\LICC,\LQICC,\SI,\SQI,\QIP\}$.
Let $\rho_{AB}$ be a bipartite quantum state and $\varepsilon\in[0,1]$. It holds that
\begin{align}
  C_{d,\cF}^{\ve}(\rho_{AB}) = \widetilde{\ell}_{\cF}^\varepsilon(\rho_{AB}).
\end{align}
\end{theorem}
\begin{proof}
``$\leq$'': This direction can be shown using similar argument 
in Lemma~\ref{lemma:equivalence-converse}.

``$\geq$'': Let $\ell_{\cF}^\varepsilon(\rho_{AB})=\log\vert L\vert$ and let $[\Lambda,\Delta,f]$ achieves this rate
such that $\Lambda_{AB\to C}\in\cF$. The key observation is that an assisted incoherent randomness extraction protocol
$[\Lambda_{AB\to C},\Delta_C,f]$ for state $\rho_{AB}$ can be regarded as an incoherent randomness extraction protocol
$[\id_{C\to C},\Delta_C,f]$ (without assistance) for the state $\tau_{C}\equiv\Lambda_{AB\to C}(\rho_{AB})$, with
purification $\tau[\Lambda]_{CER}$ defined in~\eqref{eq:rho-Lambda-CER} and joint reference system $ER$ (cf. the shaded
area of Figure~\ref{fig:assisted-extraction-alt}). That is, for the state $\tau_C$, there exists a hash function $f$
such that
\begin{align}\label{eq:gPedLUEcrgSIdg}
  d_{sec}\left(\tau[\id,\Delta,f]_{LER}|ER\right)\leq \varepsilon.
\end{align}
Recalling the one-to-one correspondence between coherence distillation protocol and incoherent randomness extraction
protocol in Proposition~\ref{prop: IO protocol}, we conclude from~\eqref{eq:gPedLUEcrgSIdg} that there exists a quantum
operation $\Gamma_{C\to L}\in\DIIO$ such that $P(\Gamma_{C\to L}(\tau_C), \Psi_L)\leq\varepsilon$. Compositing these
two quantum operations yields
\begin{align}
  P(\Gamma_{C\to L}(\tau_C),\Psi_L)
= P(\Gamma_{C\to L}\circ\Lambda_{AB\to C}(\rho_{AB}),\Psi_L) \leq\varepsilon.
\end{align}
That is to say, the composite operation $\Gamma\circ\Lambda$ distills a MCS of rank $\vert L\vert$ such that the error
is bounded by $\varepsilon$. It then suffices to show that $\Gamma\circ\Lambda\in\cF$. In the following, we handle
case by case to show that the conditions $\Lambda\in\cF$ and $\Gamma\in\DIIO$ together indeed imply that
$\Gamma\circ\Lambda\in\cF$.

\textit{Cases that $\cF\equiv\LICC$ or $\cF\equiv\LQICC$.} We only consider $\cF\equiv\LQICC$ since the other case
$\cF\equiv\LICC$ can be shown in a similar way. We focus on the last round of classical communication in the local
incoherent operations and classical communication operation $\Lambda_{AB\to C}\in\LQICC$. It must be that Bob sends
her outcome to Bob. More specifically, in the last round, Bob performs a POVM $\{E^x_A\}$ such that $E^x_A\geq 0$ and
$\sum_xE^x_A=\1$. She sends the outcome $x$ to Bob. Conditioned on $x$, Bob performs an operation $\cD^x_{B\to
C}\in\MIO$ on the post-measurement state. After that, he continues to do the operation $\Gamma_{C\to L}$. Since
$\Gamma_{C\to L}\in\DIIO$, we know $\Gamma\circ\cD^x\in\MIO$ as $\DIIO\subseteq\MIO$ and $\MIO$ is closed under
composition. As one can tell, the only difference between $\Lambda$ and $\Gamma\circ\Lambda$ is that in the last round
Bob performs different conditional operations -- for the former $\cD^x$ is taken, while for the latter
$\Gamma\circ\cD^x$ is adopted. They are both free local incoherent operations. Hence, $\Gamma\circ\Lambda\in\LQICC$.

\textit{Cases that $\cF\equiv\SI$ or $\cF\equiv\SQI$.}
We only consider $\cF\equiv\SI$ since the other case
$\cF\equiv\SQI$ can be shown in a similar way.
Assume the incoherent operations $\Lambda_{AB\to C}\in\SI$ has the form
\begin{align}
    \Lambda_{AB\to C}(\cdot) = \sum_i (A_i\ox B_i) (\cdot) (A_i\ox B_i)^\dagger,
\end{align}
where both $A_i$ and $B_i$ are incoherent Kraus operators. Since $\DIIO\subset\IO$, $\Gamma_{C\to L}$ admits an
incoherent Kraus decomposition as $\Gamma_{C\to L}(\rho)=\sum_lK_l\rho K_l^\dagger$, where $K_l$ are incoherent Kraus
operators. Thus, $\Gamma\circ\Lambda$ admits the decomposition
\begin{align}
    \Gamma_{C\to L}\circ\Lambda_{AB\to C}(\rho_{AB})
&=  \sum_l K_l \left(\sum_i(A_i\ox B_i)\rho (A_i\ox B_i)^\dagger\right) K_l^\dagger \\
&=  \sum_{l,i} \left(A_i\ox K_lB_i\right)\rho_{AB}\left(A_i\ox K_lB_i\right)^\dagger.
\end{align}
The completeness condition holds since
\begin{align}
  \sum_{l,i}\left(A_i\ox K_lB_i\right)^\dagger\left(A_i\ox K_lB_i\right)
&= \sum_i(A_i\ox B_i)^\dagger \left(\sum_lK_l^\dagger K_l\right) (A_i\ox B_i) \\
&= \sum_i(A_i\ox B_i)^\dagger(A_i\ox B_i) = \1_{AB}.
\end{align}
What's more, the Kraus operators $K_lB_i$ are incoherent since incoherent Kraus operators are closed under composition.
Hence $\Gamma\circ\Lambda\in\SI$.

\textit{Cases that $\cF\equiv\QIP$.} Notice that $\DIIO$ preserves the free incoherent states $\cI$. On the other hand
it holds by definition that $\cI\subset\mathcal{QI}$. Thus $\DIIO$ preserves $\mathcal{QI}$. This implies that
$\Gamma(\mathcal{QI})\subset\mathcal{QI}$. Since $\Lambda\in\QIP$ and $\QIP$ is closed under composition, we conclude
$\Gamma\circ\Lambda\in\QIP$.
\end{proof}

\subsection{Comparison among various assisted tasks}

We have considered three assisted tasks concerning a bipartite quantum state $\rho_{AB}$ in
this work:
\begin{enumerate}
  \item The assisted coherence distillation task introduced in Section~\ref{sec:assisted coherence distillation}
    and the corresponding one-shot optimal rate $C_{d,\cF}^{\ve}(\rho_{AB})$~\eqref{eq:one-shot-EAD},
  \item The assisted incoherent randomness extraction task introduced in
        Section~\ref{sec:assisted randomness extraction}
        and the corresponding one-shot optimal rate $\ell_{\cF}^\varepsilon(\rho_{AB})$~\eqref{eq:one-shot-er}, and
  \item An alternative formulation of the assisted incoherent randomness extraction task introduced in
    Appendix~\ref{sec:alternative formulation} and the corresponding one-shot optimal rate
    $\widetilde{\ell}_{\cF}^\varepsilon(\rho_{AB})$~\eqref{eq:one-shot-er-alt}.
\end{enumerate}
The following relation holds among these quantities as a consequence of
Proposition~\ref{prop:two-er-relation} and Theorem~\ref{thm:equivalence-alternative}:
\begin{align}\label{eq:teQjuJHcTSsXHoDTtC}
  C_{d,\cF}^{\ve}(\rho_{AB}) = \widetilde{\ell}_{\cF}^\varepsilon(\rho_{AB}) \leq \ell_{\cF}^\varepsilon(\rho_{AB}).
\end{align}
Specially, for the class of $\QIP$ we obtain an equivalence relation due to Theorem~\ref{thm:equivalence} and
Theorem~\ref{thm:equivalence-alternative}:
\begin{align}
    C_{d,\QIP}^{\ve}(\rho_{AB}) = \widetilde{\ell}_{\QIP}^\varepsilon(\rho_{AB}) = \ell_{\QIP}^\varepsilon(\rho_{AB}).
\end{align}
To prove or disprove that the inequality in~\eqref{eq:teQjuJHcTSsXHoDTtC} is an equality for other free
classes is left as future work.

\end{appendices}

\end{document}